\documentclass[superscriptaddress,aps,pra,amssymb,amsfonts,letterpaper]{revtex4}
\pdfoutput=1
\usepackage{enumerate}
\usepackage{verbatim}
\usepackage{amsmath}
\usepackage{latexsym}
\usepackage{revsymb}
\usepackage{yfonts}
\usepackage{amsfonts}
\usepackage{amsmath}
\usepackage{amssymb}
\usepackage{amsthm}
\usepackage{graphicx}
\usepackage{bm}
\usepackage{bbm}
\usepackage{color}
\usepackage{mathrsfs}
\usepackage{epstopdf}
\def\toto{\mathop{\to}\limits^{TO}}

\def\<{\langle}
\def\>{\rangle}
\def\ot{\otimes}

\def\rhor{\rho_R}

\def\rhoinr{\rho^{\rm in}_R}
\def\rhoins{\rho^{\rm in}_S}
\def\rhoinw{\rho^{\rm in}_W}
\def\rhoinc{\rho^{\rm in}_C}

\def\rhoouts{\rho^{\rm out}_S}
\def\rhooutw{\rho^{\rm out}_W}

\def\rhooutc{\rho^{\rm out}_C}

\def\sgibbs{\rho_S^\beta}

\def\cgibbs{\rho_C^\beta}
\def\ancgibbs{\rho_{\rm anc}^\beta}

\newcommand{\be}{\begin{eqnarray} \begin{aligned}}
\newcommand{\ee}{\end{aligned} \end{eqnarray} }
\newcommand{\benn}{\begin{eqnarray*} \begin{aligned}}
\newcommand{\eenn}{\end{aligned} \end{eqnarray*} }

\newcommand{\ben}{\begin{eqnarray} \begin{aligned}}
\newcommand{\een}{\end{aligned} \end{eqnarray} }

\newcommand{\h}{\mathrm{H}}
\newcommand{\bc}{\begin{center}}
\newcommand{\ec}{\end{center}}
\newcommand{\half}{\frac{1}{2}}

\newcommand{\id}{\mathbb{I}}

\newcommand{\tr}{\mathop{\mathsf{tr}}\nolimits}

\newcommand{\e}{\mathrm{e}}

\newcommand{\beq}{\begin{eqnarray} \begin{aligned}}
\newcommand{\eeq}{\end{aligned} \end{eqnarray} }
\newcommand{\bea}{\begin{array}}
\newcommand{\eea}{\end{array}}

\newcommand{\bee}{\begin{enumerate}}
\newcommand{\eee}{\end{enumerate}}
\newcommand{\bei}{\begin{itemize}}
\newcommand{\eei}{\end{itemize}}
\newtheorem{theorem}{Theorem}
\newtheorem{proposition}[theorem]{Proposition}

\newtheorem{lemma}[theorem]{Lemma}

\newtheorem{definition}[theorem]{Definition}
\newtheorem{remark}[theorem]{Remark}
\newtheorem{corollary}[theorem]{Corollary}

\usepackage{amsfonts}

\def\id{\mathbb{I}}

\def\01{\{0,1\}}
\newcommand{\ceil}[1]{\lceil{#1}\rceil}

\newcommand{\ket}[1]{|#1\rangle}
\newcommand{\bra}[1]{\langle#1|}

\newcommand{\proj}[1]{|#1\rangle\langle#1|}

\newcommand{\rank}{\operatorname{rank}}

\newtheorem{fact}{Fact}

\newcommand{\supl}{Supplementary Information}

\def\<{\langle}
\def\>{\rangle}
\def\ot{\otimes}

\def\rhor{\rho_R}

\def\gibbs{\rho_\beta}
\def\final{\rho'}

\def\s{\,\,\,\,}

\def\initial{\rho}

\newcommand{\alfree}[1]{F_\alpha(#1,\gibbs)}
\newcommand{\qalfree}{{\hat F}}
\newcommand{\qalfreesimple}{{\tilde F}}
\def\fmin{F_{\rm min}}
\def\fmax{F_{\rm max}}

\def\ep{\epsilon}
\newcommand{\sgn}{\operatorname{sgn}}
\def\qrenyi{S}
\def\qrenyisimple{\tilde S}
\def\trumpd{\rm D_{work}}

\def\expbound{\exp(-\Omega(\sqrt{\log(N)}))} 

\def\gibbsin{\rho_\beta^{(0)}}
\def\gibbsout{\rho_\beta^{(1)}}

\def\genFs{generalized free energies}
\newtheorem*{rep@theorem}{\rep@title}
\newcommand{\newreptheorem}[2]{%
\newenvironment{rep#1}[1]{%
 \def\rep@title{#2 \ref{##1} (restatement)}%
 \begin{rep@theorem}}%
 {\end{rep@theorem}}}
\makeatother

\newreptheorem{thm}{Theorem}
\newreptheorem{lem}{Lemma}

\begin{document}
\title{The second laws of quantum thermodynamics}
\author{Fernando G.S.L. \surname{Brand\~ao}}
\affiliation{University College London, Department of Computer Science}
\author{Micha\l\ \surname{Horodecki}}
\affiliation{IFTIA, University of Gda\'{n}sk, 80-952 Gda\'{n}sk, Poland}
\author{Nelly Huei Ying \surname{Ng}}
\affiliation{Centre for Quantum Technologies, National University of Singapore, 3 Science Drive 2, 117543 Singapore}
\author{Jonathan \surname{Oppenheim}}
\affiliation{University College of London, Department of Physics \& Astronomy, London, WC1E 6BT and London Interdisciplinary Network for Quantum Science}                        
\affiliation{Centre for Quantum Technologies, National University of Singapore, 3 Science Drive 2, 117543 Singapore}
\author{Stephanie \surname{Wehner}}
\affiliation{Centre for Quantum Technologies, National University of Singapore, 3 Science Drive 2, 117543 Singapore}
\affiliation{School of Computing, National University of Singapore, 13 Computing Drive, 117417 Singapore}

\begin{abstract}
The second law of thermodynamics tells us which state transformations
are so statistically unlikely that they are effectively forbidden. Its
original formulation, due to Clausius, states that ``Heat can never pass
from a colder to a warmer body  without some other change, connected
therewith, occurring at the same time'' \cite{clausius1850ueber}. The second
law applies to systems composed of many particles interacting;
however, we are seeing that one can make sense of thermodynamics in the regime where we
only have a small number of particles interacting with a heat 
bath~\cite{uniqueinfo,dahlsten2011inadequacy, del2011thermodynamic,HO-limitations,aaberg-singleshot,
egloff2012laws,skrzypczyk2013extracting,faist2012quantitative}. 
Is there a second law of thermodynamics in this
regime? Here, we find that for processes which are cyclic or very close to cyclic,
the second law for microscopic systems
takes on a
very different form than it does at the macroscopic scale, imposing not
just one constraint on what state transformations are possible, but an
entire family of constraints. In particular, we find a family of free energies
which generalise the traditional one, and show that they can never increase.
The ordinary second law
just corresponds to the non-increasing of one of these
free energies, with the remainder imposing additional constraints on thermodynamic transitions.
We further find that there are three regimes which determine which family
of second laws govern state transitions, depending on how cyclic the
process is. In one regime one can cause an apparent violation 
of the usual second law, through a process of embezzling work
from a large system which remains arbitrarily close to its original state.
These second laws are not only relevant for small systems, 
but also apply to individual macroscopic systems interacting via long-range interactions,
which only satisfy the ordinary second law on average. 
By making precise the definition of thermal operations,
the laws of thermodynamics take on a simple form with the first law defining the 
class of thermal operations, the zeroeth law, as derived here, emerging as a unique condition ensuring the
theory is nontrivial, and the remaining laws being a monotonicity property of our generalised free energies.
\end{abstract}
\maketitle

In attempting to apply the Clausius statement of the second law to the microscopic or quantum scale, we immediately run into a problem, 
because it talks about cyclic processes in which there is\textit{ no other change} occurring at the same time, and at this scale, 
it is impossible to design a process in which
there is no change, however slight in our devices and heat engines. Interpreted strictly, the Clausius statement of the second law,
applies to situations which never occur in nature. The same holds true for other versions of the second law, such as the Kelvin-Planck statement, where 
one also talks about cyclic processes, in which all other objects beside the system of interest are returned back to their original state. Here,
we derive a quantum version of the Clausius statement, by looking at 
processes where a microscopic or quantum system undergoes a transition from one state to another, 
while the environment, and working body or heat engine is returned back to their original state.
While macroscopically, only a single second law restricts transitions, we find that there are an entire family of more fundamental 
restrictions at the quantum level. At the macroscopic scale, and for systems with short range correlations, 
this entire family of second laws become equal to the ordinary second law, but outside of this regime, these other second laws impose additional
restrictions on thermodynamical transitions. What's more, one needs to be more precise about what one means by a cyclic process.
At the macroscopic scale, the fact that a process is only approximately cyclic has generally been assumed to be enough to guarantee the second law.  
Here, we show that this is not the case in the microscopic regime, and we therefore needs to talk about ``how cyclic'' a process is when stating 
the second law. We also derive in this work, a zeroeth law of thermodynamics, which is stronger than the ordinary zeroeth law.

For thermodynamics at the macroscopic scale, a system in state $\initial$ can be transformed into state $\final$ provided that the free energy goes
down, where the free energy for a state $\rho$ is
\begin{align}
F(\rho) = \<E(\rho)\> - kTS(\rho),
\label{eq:helmfreeenergy}
\end{align}
with $T$ the temperature of the ambient heat bath that surrounds the system, $k$, the Boltzmann constant, $S(\rho)$ the entropy of the system, and $\<E\>$ 
its average energy. This is a version of the second law, where we also use the fact that the total energy of the system and heat bath must
be conserved. This criteria governing state transitions is valid if the system is composed of many particles, and there are no long range correlations. 
In the case of microscopic, quantum or highly correlated systems, a criteria for state transitions of a total system was proven in
\cite{HO-limitations} and named thermo-majorisation (See Figure \ref{fig:thermomaj}).
This criterion has been conjectured~\cite{ruch1975diagram} and claimed 
to be~\cite{egloff2012laws} a second law. 
However, here, we will see that if elevated to such high status, it can be violated. Namely, we will give examples where
$\initial\rightarrow\final$ would violate the thermo-majorisation criteria, but nonetheless, the transition is possible via
a cyclic process in which a working body $\sigma$ - an ancilla or catalyst - is returned back into its original state.

This phenomenon is related to entanglement catalysis~\cite{JonathanP}, where it can be shown
that some forbidden transitions can be possible, if we can use an additional system $\sigma$ as a catalyst, i.e. 
we may have $\initial \not\rightarrow \final$ and yet $\initial \otimes\sigma\rightarrow\final\otimes\sigma$.
In the case of thermodynamics, the catalyst $\sigma$ may be thought of as a working body or heat engine
which undergoes a cyclic process and is returned back into its original state. In deciding whether
one can transform $\initial$ into $\final$, one therefore needs to ask whether there exists a working body
or other ancillas $\sigma$ for which $\initial \ot\sigma \rightarrow \final \ot \sigma$ (see Figure \ref{fig:setting}).  Thus, thermomajorisation
should only be applied to total resources including catalysts and working bodies and not the system of interest itself.
In the case of entanglement theory,
and when the catalyst is returned in exactly the same state, the criteria for when one pure state
may be transformed into another has been found~\cite{Klimesh-trumping,Turgut-trumping-2007} and they are called trumping conditions. We will generalise and adapt the trumping conditions to enable their application to the case of thermodynamics.

\begin{figure}
\includegraphics[width=10cm]{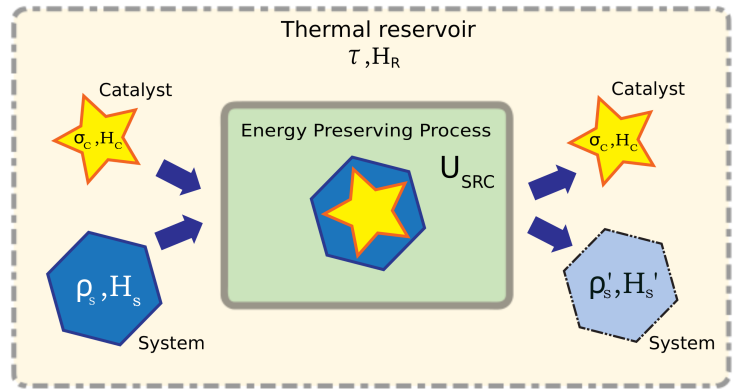}
	\caption{In the microregime, when can a state $\rho_S$ with Hamiltonian $H_S$ be transformed to 
	a state $\rho_S'$ and Hamiltonian $H_S'$? In order to do so, one can couple the system to a heat bath $\gibbs=e^{-\beta H_R}/Z$ 
with Hamiltonian $H_R$ 
use any devices as long as they are returned back in their original state 
(thus we may think of them as a catalyst - $\sigma$) and we are allowed to perform any action as long as we preserve the 
overall energy (see below for a more detailed description of these operations, which we call {\it catalytic thermal operations}). 
Loosely speaking, our second law says that $\rho_S$ can transition to $\rho_S'$ if and only if $\rho_S'$ is closer to the 
thermal state $\gibbs$ of the system at inverse temperature $\beta$ 
with respect to \emph{all} R{\'e}nyi-divergences. In the thermodynamic limit, 
all these quantities converge so that we recover the usual second law.}
\label{fig:setting}
\end{figure}

\begin{figure}
  \centering
  \includegraphics[width=15cm]{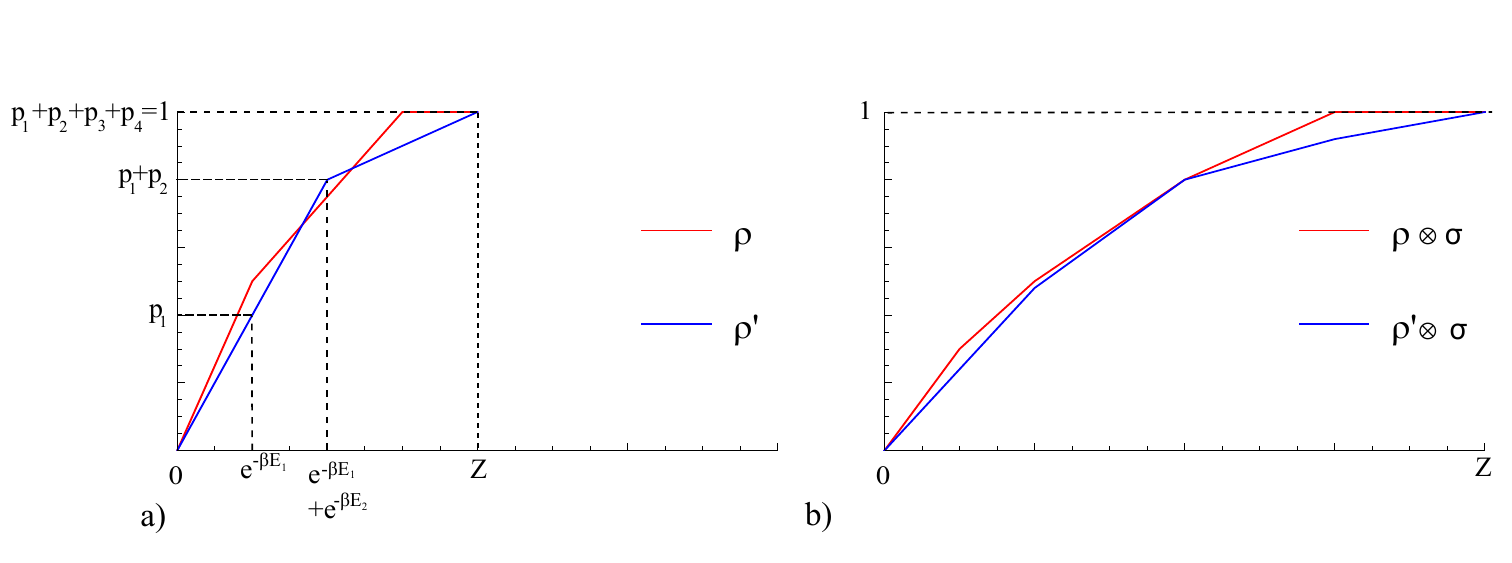}
\caption{The thermo-majorisation criteria is as follows: Consider probabilities $p(E,g)$ of the initial system $\initial$ to be in the $g$'th state of energy $E$.  
Now let us
put $ p(E,g)\e^{\beta E}$ in decreasing order $p(E_1,g_1)\e^{\beta E_1}  \geq  p(E_2,g_2)\e^{\beta E_2}\geq p(E_3,g_3)\e^{\beta E_3} ...$ -- we say
that the eigenvalues are $\beta$-ordered.  We can do the same for system $\sigma$ 
i.e. $\e^{\beta E_1} q(E_1,g_1) \geq \e^{\beta E_2} q(E_2,g_2)\geq\e^{\beta E_3} q(E_3,g_3)...$.  Then the condition which determines whether
 we can transform $\initial$ into $\final$ is depicted in the above figure.  Namely, for any state, we construct a curve with points $k$ given
 by $\{\sum e^{-\beta E_i},\sum_i^k p_i \}$. Then a thermodynamical transition from $\initial$ to $\final$ is possible if and only if, the curve of $\initial$ 
lies above the curve of $\final$.  One can make a previously impossible transition possible by adding work in the form of the pure state $\psi_W$
which will scale each point by an amount $e^{-\beta W}$ horizontally~\cite{HO-limitations}. The above
criteria should not be applied to the system of interest alone, but the system plus any additional resources which are used to enable the transition. In the example above,
figure (a) we cannot transform $\rho$ into $\rho'$ without performing work (or visa-versa). However, by using a resource $\sigma$ which is returned in it's initial state, we see in (b) that 
$\rho\otimes\sigma$ thermo-majorizes $\rho'\otimes\sigma$ so that we can make a thermodynamical transition from $\rho$ to $\rho'$ without adding additional work. }
\label{fig:thermomaj}
\end{figure}

\subsubsection*{A family of second laws}

Here we consider all possible cyclic thermodynamical processes, and show 
that transition laws are affected by using ancillary systems which are returned back to their initial state. 
Rather than a single free energy which determines which transitions are possible, we find
necessary 
and sufficient conditions 
for thermodynamic transitions 
which form, not just one, but a family of second laws.  We define the free energies
\be
F_\alpha(\rho,\gibbs):=kTD_\alpha(\initial\|\gibbs)-kT\log{Z},
\label{eq:genfree}
\ee
with the R{\'e}nyi divergences $D_\alpha(\initial\|\gibbs)$ defined as 
\be
D_\alpha(\initial\|\gibbs)= \frac{\sgn(\alpha)}{\alpha-1} \log \sum_i p_i^\alpha q_i^{1-\alpha},
\label{eq:renyidivergence}
\ee
where $p_i$, $q_i$ are the eigenvalues of $\initial$ and the thermal state of the system is $\gibbs=\sum_{i,g} e^{-\beta H_S}/Z$ with Hamiltonian $H_S$, partition function
$Z=\sum_{i,g} e^{-\beta E_i}$  and $\beta = 1/T$ with $T$ the temperature of the surrounding heat bath.

We can then state quantum second laws, and ones that hold for states block diagonal in the 
energy basis. In the latter case, we find the following set of second laws:

\smallskip
\noindent
{\it In the presence of a single heat bath, the free energies
$F_\alpha(\initial,\gibbs)$ 
do not increase for $\alpha\geq 0$
That is, $\forall \alpha\geq 0$, $F_{\alpha}(\initial,\gibbs) \geq F_{\alpha}(\final,\gibbs)$ where 
$\initial$ and $\final$ are the initial and final state. Moreover, if $F_{\alpha}(\initial,\gibbs) \geq F_{\alpha}(\final,\gibbs)$ holds $\forall\alpha \geq 0$, then there exists a catalytic thermal operation that transforms $\rho$ to $\rho'$.} 

\smallskip
\noindent
We say that $F_\alpha(\initial,\gibbs)$ are {\it monotones} -- the system always gets closer to the thermal state, 
thus the function always decreases. 
By including an auxiliary system as described in \cite{HO-limitations}, the above statement of the second law is equivalent to
the case where one changes the Hamiltonian of the system, in which case, one could write $F_{\alpha}(\initial,\gibbs) \geq F_{\alpha}(\final,\gibbs')$
where the initial Hamiltonian $H_S$ has been changed via external control to the final Hamiltonian $H_S'$, with $\gibbs$ and $\gibbs'$ being the respective thermal states. This is described in Section 
\ref{ss:changingham} of the \supl. Note that in fact $F_{\alpha}(\initial,\gibbs)$ is a monotone for all $\alpha\in (-\infty,\infty)$
but since we are allowed to borrow a pure state and return it in a state arbitrarily close to its initial state, only $\alpha\geq 0$ are relevant, as can be seen by noting that if any of the probabilities $p_i$ in Equation \eqref{eq:renyidivergence} are zero, then for $\alpha<0$, $F_\alpha(\initial,\gibbs)$ diverges and will thus always be monotonic.

These set of limitations are less stringent than thermo-majorisation. 
Not only do these second laws provide limitations, but they are also sufficient -- 
whenever the free energies of one state are all greater than for the other state,
one can transform the one into the other one. We prove this in Sections \ref{sec:secondlaws} and 
\ref{sec:adultHamil} of the \supl. Note that the monotonicity of \eqref{eq:genfree} establishes a continuous family of 
conditions, one for each value of $\alpha$.
However, in the case of larger systems, one can perform a quick check, namely:
we find that for any distribution $p$ we can construct \emph{smoothed} distributions that are very close to $p$, and in terms of these smoothed distributions, check two conditions in terms of the two free energies for $\alpha=0,\infty$ found in~\cite{HO-limitations}. If such conditions are satisfied on the smoothed distribution, it implies that the infinite set of conditions is satisfied as well.

For $\alpha\rightarrow 1$, 
$\alfree{\initial}$ is equal to the ordinary Helmholtz free energy $F(\initial)$,
hence our conditions include the ordinary second law (combined with energy conservation), 
and we thus see that it is merely one of many constraints on thermodynamical state transitions.
In the macroscopic regime, and for systems which are not highly correlated, $F_\alpha(\initial,\gibbs)\approx F_1(\initial,\gibbs)$ for all $\alpha$
which explains why the single constraint given by the usual second law is more or less adequate in this limit. 
For $\alpha=0$, $F_0(\initial,\gibbs)=\fmin(\initial)$, which we previously found
to quantify the maximal amount of work extractable from a system in contact with a reservoir under all thermal operations~\cite{HO-limitations} or
in a model of alternating adiabatic and isothermal operations~\cite{aaberg-singleshot}. In the case of trivial Hamiltonians, this quantity is 
a strengthened version of that found previously in \cite{dahlsten2011inadequacy}
in that we have tight necessary and sufficient conditions. 
Generally, we find that although generic state transitions are affected by catalysts, the results in \cite{HO-limitations} on distillable work 
is not. Likewise,
the reverse process, the so-called {\it work of formation} $\fmax(\initial)$~\cite{HO-limitations} in creating $\rho$ corresponds to $\alpha\rightarrow \infty$. We thus see that
the Helmholtz free energy and two free energies proposed in~\cite{HO-limitations} are special cases of our family of conditions and they 
hold even in the presence of catalysts.

However in other cases, we find that one can distill more work than would be obtainable without considering an ancillary system which is used as a catalyst. 
As a simple application of our results, we find that by using an ancillary system, one can erase or reset a memory register at a lower work cost than previously known. 
In particular, this can occur when
resetting a memory register to a pure state, while retaining
correlations with a reference system~\cite{del2011thermodynamic,faist2012quantitative}. Classically, resetting a memory requires work, but 
the authors of~\cite{del2011thermodynamic} found that due to entanglement, there were cases where a memory could be reset at a cost
of a negative amount of work (i.e. work could actually be extracted while the memory was reset). Here, we find that even more work
can be extracted during a memory reset.
In general, we find 
that more work can be extracted in a cyclic process, and in Section \ref{ss:erasure} of the \supl, we derive the optimal amount of extractable work in such a case, providing an operational interpretation for the difference of R{\'e}nyi entropies, and an interpretation for when this quantity is negative.

For states that are not diagonal in the energy basis, we provide generalization 
of the above limitations in terms of quantum alpha-free energies in Sec. \ref{sec:secondlaws} of \supl. 
These form a family of fully quantum second laws, which are necessary, but not sufficient conditions for state transformations. 
Due to the non-commutative nature of the state of the system and
the thermal state, our new free energies have a more complicated form and are based on quantum 
Renyi divergences \cite{HiaiMPB2010-f-divergences,Muller-LennertDSFT2013-Renyi,WildeWY2013-strong-converse} (see also  \cite{JaksicOPP2012-entropy}).
Defining two quantum versions of $F_\alpha(\rho,\gibbs)$
\be
\qalfreesimple_\alpha(\rho,\gibbs)=kT \frac{{\rm sgn}(\alpha)}{\alpha-1} \log\tr \rho^\alpha \gibbs^{1-\alpha}-kT\log{Z}
\ee
and 
\be
\qalfree_\alpha(\rho,\gibbs)=kT\frac{1}{\alpha-1}\log 
\left(\tr( \gibbs^{\frac{1-\alpha}{2\alpha}} \rho \gibbs^{\frac{1-\alpha}{2\alpha}})^\alpha \right) -kT\log{Z}
\ee
We are able to find
\begin{itemize}
\item{\bf Quantum second laws:} {\it A transition from $\initial$ to $\final$ is possible, only if
\be
\qalfree_\alpha(\rho,\gibbs) \geq \qalfree_\alpha(\final,\gibbs)
\nonumber
\ee
for $\alpha\geq \frac12$ and 
\be
\qalfree_\alpha(\gibbs , \rho) \geq \qalfree_\alpha(\gibbs,\rho)
\nonumber
\ee
for  $\frac12 \leq \alpha\leq 1$ and 
\be
\qalfreesimple_\alpha(\initial, \gibbs) \geq \qalfreesimple_\alpha(\final,\gibbs)
\nonumber
\ee
for  $0 \leq \alpha\leq 2$.
}
\end{itemize}
where once again the above laws include transitions where the Hamiltonian changes by making use of an ancillary system as is done in 
\cite{HO-limitations}.

\subsubsection*{Work distance}

Given the monotonicity of $F_\alpha(\initial,\gibbs)$ we may easily compute the maximum amount of deterministic work which can be extracted when going from 
a system in state $\initial$ to one in state $\final$. Namely, in \cite{HO-limitations} we introduced the notion of a work bit, or {\it wit}, 
which starts off in state $|0\rangle$ and gets raised or lowered to a state $|W\rangle$ with energy $W$. This corresponds to extracting
an amount of work $W$ if $W$ is positive, or performing work if $W$ is negative.  Since our second laws concern general state transformations, 
they can be applied to the case of deterministic work extraction or to the case of extracting a fluctuating amount of work, but here we apply our second
laws to the former case.

From them, we know that a transition is possible if and only if
\begin{align}
F_\alpha(\initial \otimes\proj{0},\gibbs) \geq F_\alpha(\final\otimes\proj{W},\gibbs) \s \forall \alpha \geq 0
\end{align}
which implies (see Section \ref{sec:inexact}) of the \supl), that $W= \trumpd(\rho\succ\rho')$ is achievable, where
\begin{align}
\trumpd(\rho\succ\final):= kT\inf_{\alpha} [F_\alpha(\rho,\gibbs)-F_\alpha(\final,\gibbs)].
\label{eq:trumpd}
\end{align}
We thus see that the $F_\alpha(\rho,\gibbs)$ are very much like free energies, not only in the sense that they are monotones, but 
also in the sense that the amount of work is given by the function's difference between the initial and final state (albeit for the minimal one).
The quantity on the right hand side of Equation \eqref{eq:trumpd} can also be thought of as a distance measure between states, as was done with the thermo-majorisation criteria in \cite{egloff2012laws}  
 and we will henceforth refer to it as the \textit{work distance} from $\initial$ to $\final$.

\subsubsection*{Approximately cyclic processes}
As in the context of entanglement theory~\cite{BBPS1996,BBPSSW1996,thermo-ent2002,DevetakHW2005-resource}, 
we consider thermodynamics as a resource theory \cite{Beth-thermo,uniqueinfo,HO-limitations,thermoiid},  
where we are allowed to implement a class of operations, and then quantify 
the resources which cannot be created under the class of operations. For thermodynamics, various classes of operations have been considered in the micro-regime
\cite{Beth-thermo,Alicki79,allahverdyan2000extraction,feldmann2006lubrication,linden2010small,HO-limitations,aaberg-singleshot,gemmer2009quantum,hovhannisyan2013role,AlickiFannes2012-battery,GalbwaserAK2013-work}. 
In particular, we consider {\it thermal operations} \cite{HO-limitations,thermoiid,Beth-thermo} where we allow the systems
of interest $\initial$ to be coupled to a thermal reservoir $\rhor$ in the thermal state at temperature $T$,
 and we allow arbitrary unitaries between the system, working body and reservoir 
which conserve energy. Energy conservation is important, because we need to account for all sources of energy which might get added to our system. 
Operations which pump energy into or from the system can be incorporated into the paradigm by including the source of energy as 
an ancillary system. This paradigm is equivalent to ones in which we allow interaction
Hamiltonians rather than unitaries, or where we allow for a Hamiltonian which changes with time, provided that all ancillas are carefully accounted for~\cite{thermoiid}. 

In this article we consider a natural additional ingredient, namely, we allow additional resources such as a heat engine and working body (or ancilla system) $C$ in state $\rhoinc$ which must be returned in its initial form. Considering such resources are crucial, since the experimenter who is trying to extract work or otherwise manipulate a system, should be allowed as much ingenuity as possible. However, to ensure that work is not being added to the system from the heat engine itself, we demand that the thermodynamic process be cyclic, in the sense that the additional resources are returned in state $\rhooutc$, which is approximately equal to its original state $\rhoinc$. In essence, how cyclic a process is can be understood in terms of how good an approximation $\rhoinc$ is to $\rhooutc$.  In the macroscopic case, the notion of how cyclic a process is, has not been deemed important -- it was assumed that it was enough to take the initial and final state of the heat engine as being close to the same macroscopic state. We find that this is not the case, even if one demands that the initial and final states are arbitrarily close. This is increasingly important when we are manipulating microscopic systems, where small differences can have more noticeable effects. We find that
depending on the desired approximation, i.e., depending on how cyclic we demand the process to be, there are several different regimes of second laws.

The simplest case is where $\rhoinc=\rhooutc$, that is, the process is perfectly cyclic and the catalyst is restored to its original form. In this setting, we have the second law as stated above. But no real process is perfectly cyclic, and so, it is important to consider the case where  $\rhoinc\approx\rhooutc$. 
This requires us to derive approximate transformation conditions which we expect to also find application in entanglement theory, and are contained in Section \ref{sec:inexact} of the \supl.
We find that the form the second law takes when the process is not perfectly cyclic, depends on 
the degree to which our process is cyclic.
We find three separate regimes
which are quantified by how cyclic the process is, in terms of how close $\rhoinc$ is to $\rhooutc$. 

In the first regime, we demand that the change in the working body through a cycle is small, in the sense that $\trumpd(\rhooutc\succ\rhoinc)\leq \epsilon$. 
In other words, any change in the working body could be corrected by applying a small amount of work. In this case, we recover the second laws as stated above. 

The second regime is when the change in the working body has error inversely proportional to the
number of particles it has,  i.e. $\|\rhoinc-\rhooutc\|_1\leq \epsilon/\log N$, where $N$ is the dimension of the catalyst. In this case, we retrieve the
standard second law. The ordinary free energy continues to govern whether a thermodynamical transition is possible, while the R{\'e}nyi-divergences do not. 
We thus see that the ordinary second law can arise 
in the macroscopic limit, or if we allow processes which deviate from being cyclic in a manner which is constant per number of particles in the working body. We detail this in Section \ref{sec:extensive_error} of the \supl\ by showing that when the standard second law holds, we can construct a catalyst to enable the transition, while being returned with small error per particle.

Finally, we consider the regime where we simply 
demand that the process is close to cyclic regardless of the size of the ancillary system. i.e. $\|\rhoinc-\rhooutc\|_1\leq \epsilon$. Since $\epsilon$ can be arbitrarily small,
one would imagine that for such an approximately cyclic process one recovers a second law of some sort. Nonetheless, we find that for any $\epsilon$, no matter how small, one can construct a working body, and cyclic
process, such that one can pump heat from a cold reservoir to a hot reservoir, in violation of the Clausius statement of the second law. In fact, we can make arbitrary state transformations by taking the size of the working body to be so large, that work can be extracted from a single heat bath, while barely modifying the 
state of the working body. This is related to a phenomenon in entanglement theory known as 
embezzling~\cite{Hayden-embezzling}. In particular, for any desired approximation $\epsilon$ there exists a dimension $d$ such that the catalyst 
$\rhoinc = \sum_{j=1}^{d} \frac{1}{j}\proj{j}$ allows us to transform any initial state $\initial$ to any final state $\final$ such that $\|\rhoinc-\rhooutc\|_1\leq \epsilon$.

\subsubsection*{Discussion}

The second law is often seen as arising from an experimenter's lack of control over the system of interest. Here we see that this is not the case -- we obtain
our fundamental limitations even in the case where the experimenter can access the microscopic degrees of freedom of the heat bath and couple it in an arbitrary way with the system.
The reason that such fine control does not lead to a violation of the second law is related to the fact that a Maxwell's demon with microscopic control over a system cannot
violate the second law -- a demon which knows the positions and momentums of the particles of a system, must record this information in a memory, which then needs to be reset at the
end of a cyclic process~\cite{Bennett82,Landauer}. For the same reason, an ability to access the degrees of freedom of the heat bath would also require work to perform such a memory resetting step.
Remarkably, although the limitations are derived assuming that one can perform all possible operations, they are achievable using a very limited set of operations -- namely, changing the energy levels
of the system, and putting parts of  the system in thermal contact with the reservoir.
%

We have derived a family of fundamental limitations on thermodynamical state transformations for both quasi-classical states, and fully quantum states.  
Since these limitations are given in terms of generalisations of the free energy, they can be thought of as second laws, 
combined with the first law, i.e. energy conservation. For an isolated system, one could take the second law to be the increase 
in the Renyi entropies, which holds if the allowed class of operations are mixtures of unitaries. This can be thought of as resulting from 
a coarse-graining or lack of information about the full dynamics, but we do not consider in detail this approach here. 

Thinking of thermodynamics as a resource theory, allows us to re-formulate the laws of thermodynamics in a very natural way.  In essence, the zeroeth law defines the set of allowed free states (the thermal state), the first law the set of allowed operations (namely, energy conservation), and the second law is derived from these conditions to specify the set of allowed transitions. 
This has the advantage of separating out laws of fundamental physics e.g. that evolution be unitary and energy conserving, from those of thermodynamics. 
  
To state the zeroeth and first law of thermodynamics more explicitly, let us define the set of {\it catalytic thermal operations} introduced here.

\begin{itemize}
\item {\bf Catalytic thermal operations:} {\it Given a system in initial state $\rho_S$ with respect to Hamiltonian $H_S$, we can 
    \begin{enumerate}
    \item borrow a catalyst system in state $\sigma$ w.r.t. some Hamiltonian $H_C$ and returning it in a state at most $\epsilon$-close to $\sigma$. What one means by {\it close} will determine which family of second laws apply, and in the subsequent statement we invoke the most stringent conditions, namely closeness in terms of work distance,    
    \item add an arbitrary number of copies of a state $\tau$ with any Hamiltonian $H_R$,
    \item perform any unitary operation $U$ such that $[U,H_S+H_R+H_C]=0$ ({\bf The First law}), and
    \item perform partial trace over systems $R$ and $C$.
    \end{enumerate}}
\end{itemize}
Note that the demand that the unitary commutes with the total Hamiltonian implies that energy is conserved. Conversely, it is easy to see that the only process that conserves energy for an aribtrary state must commute with the Hamiltonian.  We emphasize that the equivalence of this paradigm to others has already been addressed in~\cite{thermoiid}. We here view the first law, not as something which is a consequence of thermodynamics, but rather, one which defines what thermodynamics is.
\begin{itemize}
	\item{\bf Zeroeth law:} {\it Let $\mathcal{S} = \{(\tau,H_R)\mid \tau = \gibbs^R = e^{-\beta H_R}/Z {\rm\ is\ a\ thermal\ state\ for\ } H_R\}$. 
		If $(\tau,H_R) \notin \mathcal{S}$, then arbitrary state transitions are possible and the theory becomes trivial.}
\end{itemize}

We thus see that the ordinary zeroeth law
 is replaced by the following fact: if our class of operations include energy conserving operations and the ability to add an arbitrary number of copies of some state $\tau$ corresponding to a Hamiltonian $H_R$, then the only pair of ($\tau,~H_R$) which does not make for a trivial theory (in the sense that all state transformations would be possible), is if $\tau=\gibbs$, where $\gibbs$ is the thermal state \cite{HO-limitations} with respect to $H_R$. 
This is related to the fact that the thermal state is the only state that is {\it completely passive}, i.e. one cannot draw work from arbitary number of copies of the state \cite{pusz_passive_1978}. 
%
Taking a resource perspective, we now see we have
an equivalence relation on $\mathcal{S}$ which defines for us the notion of temperature. More precisely, we will call $(\tau_1,H_{R_1})$ and $(\tau_2,H_{R_2})$ equivalent resources if and only
if no work can be gained from $\tau_1^{\otimes \ell_1} \otimes \tau_2^{\otimes \ell_2}$ for arbitrary number of copies $\ell_1$ and $\ell_2$. Thus the ordinary Zeroeth law, is replaced here by a unique condition which tells us what class of free states make thermodynamics non trivial.


Here, our derivation of the second laws is information theoretic in nature, requiring none of the assumptions usually required for 
the second law to hold. This includes ergodicity, mixing, coarse-graining of degrees of freedom and lack of control over the system. 
Monotonicity of  $F_\alpha(\rho,\gibbs)$  thus provides a powerful tool to determine what sorts of thermodynamical transitions
are possible on the quantum scale, or equally well, for systems which have long range interactions.  
From a foundational perspective, the laws of thermodynamics take on a very simple and elegant form-- a class of operations, and a set of statistical distances to the free state $\gibbs$ which can never decrease. One hopes that such information theoretic laws
can be used to discover a broad range of thermodynamical phenomena at the quantum level.

\pagebreak

\section{Supplementary Information}

In what follows
we present the full technical details of all our findings. Section \ref{preliminaria} provides the definition of the basic quantities used in our work, namely the R{\'e}nyi entropies and divergences and useful relations between them. In Section \ref{sec:trivialHamilExact}, we briefly summarize the theory of majorization and trumping, then introduce the connection between R{\'e}nyi divergences and trumping, and show how to reduce the number of trumping relations required. We then show that this provides the basis of catalytic transformations in the simple setting where the Hamiltonian is trivial. 

In Section \ref{sec:TO} and \ref{sec:adultHamil}, we consider catalytic operations in the context of full thermodynamics with general Hamiltonians. Our main technical result is Theorem\ref{th:catalysis}, which generalises trumping relations from the bistocastic case, to the case of arbitrary channels which take a subset of states to some other subset of states. We expect this to be of independent interest with applicability in a broad range of contexts. This provides the mathematical tools to formulate necessary and sufficient conditions for thermodynamic state transformations for states block diagonal in the energy eigenbasis.  These are detailed in Section \ref{sec:secondlaws}. Here, we also observe the elimination of conditions for ranges of $\alpha$, when certain states are allowed to be used and returned with close fidelity. In particular, borrowing a pure state allows us to eliminate conditions where $\alpha<0$, which significantly simplifies the conditions on state transformations. We also apply our second laws to the erasure of states while preserving a memory in Subsection \ref{ss:erasure}. In Subsection \ref{ss:qsecondlaw} we present a family of fully quantum second laws.

Section \ref{sec:inexact} makes an extensive presentation on the case where inexact catalysis is allowed, and we show that for different regimes of closeness between the catalyst's initial and final state, we obtain different second laws.  The different regimes include exact catalysis, returning the catalyst such that a small amount of work is required to bring it to it's original state, returning the catalyst with small error per particle, returning it such that only a small amount of work on average is required to bring it back to its original state, and returning it with arbitrary good fidelity, where we find that all state transformations are possible.
In \supl\ \ref{proofofprops} we show how for larger systems, one can often check just two quantities to determine if a state transformation is possible. In \ref{sec:comparison}, we discuss in more detail how our paradigm incorporates changing Hamiltonians. Since we have conditions for state transformations using ancillary systems, 
we are in a position to prove the optimality of the procedure used in \cite{HO-limitations}. The work storage system is also an ancillary system, and we further discuss the equivalence
of the work bit, to other batteries.

\appendix

\section{Preliminaria: R{\'e}nyi relative entropies and their properties}
\label{preliminaria}
\subsection{R{\'e}nyi divergence}
Consider probability distributions $p=\lbrace p_1, p_2, ..., p_n\rbrace$ and $q=\lbrace q_1, q_2, ..., q_n\rbrace$. The R{\'e}nyi divergences are defined for $\alpha \in [-\infty, \infty]$ as follows 
\be
D_\alpha(p \| q)= \frac{{\rm sgn}(\alpha)}{\alpha-1} \log\sum_i^n p_i^\alpha q_i^{1-\alpha},
\ee
where 
\beq
{\rm sgn}(\alpha)=\left \{\bea{ll}
1 &(\alpha\geq 0); \\
-1 & (\alpha < 0).
\eea\right.
\label{hfunction}
\eeq
We use the conventions $\frac00=0$  and $\frac{a}{0}=\infty$ for $a>0$. The cases $\alpha=0,1,\infty,-\infty$ are defined via the suitable limit, namely 
\beq
&D_0(p \| q)=\lim_{\alpha\to 0^+}D_\alpha(p \| q)=-\log\sum_{i:p_i\not=0}^n q_i,& \quad D_1(p \| q)=\lim_{\alpha\to 1}D_\alpha(p \| q)=\sum_i^n p_i (\log p_i - \log q_i),  \nonumber \\
& D_\infty(p \| q) = \lim_{\alpha\to \infty}D_\alpha(p \| q)=\log \max_i \frac{p_i}{q_i},& \quad D_{-\infty}(p \| q) =- \lim_{\alpha\to -\infty}D_\alpha(p \| q)=D_\infty( q \| p).
\eeq
Note that for $\alpha\rightarrow\infty$, $D_\alpha (p\|q)$ can also be expressed as:
\beq
\label{infdivergence}
D_\infty (p\|q) = \log \min \left\lbrace \lambda| \forall i, \lambda \geq \frac{p_i}{q_i} \right\rbrace
\eeq
Also, there exists a useful relation between two divergences for $\alpha\not\in\{0,1\}$:
\be
\alpha {\rm sgn}(1 - \alpha) D_{1-\alpha}(p || q)= (1-\alpha) {\rm sgn}(\alpha) D_\alpha(q || p).
\ee
\label{eq:symmetry}
For some properties of R{\'e}nyi divergence, the reader can refer to \cite{timerven2010, timerven2012}. Note that discussions in other literatures define R{\'e}nyi divergence for only non-negative alphas. 
However, the relative entropy $D_\alpha$ we define satisfies the data processing inequality for all $\alpha\in[-\infty,\infty]$:
\be
D_\alpha(\Lambda(p) \|\Lambda(q))\leq D_\alpha(p \| q),
\label{eq:data-proc}
\ee
where $\Lambda$ is a stochastic map. 
For $\alpha\in [0,\infty]$ the Renyi divergence is  nondecreasing in $\alpha$:
\be
D_\alpha(p|q)\leq D_{\gamma}(p|q)
\ee
for $\alpha\leq \gamma$ (see Theorem 3 of  Ref. \cite{timerven2010}).  

\subsection{R{\'e}nyi entropy}
The R{\'e}nyi entropies are defined for $\alpha\in R\setminus\{0,1\}$ as 
\be
H_\alpha(p)=\frac{{\rm sgn}(\alpha)}{1-\alpha} \log \sum_{i=1}^n p_i^\alpha, 
\ee
where ${\rm sgn}(\alpha)$ has been defined in \eqref{hfunction}. Again, for $\alpha \in \{-\infty,0,1,\infty\}$ we define $H_\alpha  (p)$ by taking limits. Explicitly we have 
\be
H_0(p)=\log \rank(p),\quad H_1(p)=-\sum_{i=1}^n p_i \log p_i,\quad  H_{\infty}(p)= -\log p_{\max}\quad 
H_{-\infty}(p)= \log p_{\min}
\ee 
where $\rank(p)$ is the number of nonzero elements of $p$, and  $p_{\max},~p_{\min}$ are the maximal and minimal 
element of $p$, respectively.  
The R{\'e}nyi entropies can be recovered from the relative R{\'e}nyi entropies as follows 
\be
\label{interchangeHD}
H_{\alpha}(p) = \sgn(\alpha) \log n - D_\alpha(p \| \eta)\ ,
\ee
with $\eta = \lbrace\frac{1}{n},\frac{1}{n},...,\frac{1}{n}\rbrace$ is the uniform probability distribution.

It is worth noting that the R{\'e}nyi divergences and entropies have generally been defined only for positive alphas. However, for completeness, we have generalized the definitions to negative alphas so that conditions for state transformations can be described fully by these quantities. 

\subsection{Quantum Renyi divergence}
\label{ss:quantumdivergence}

The Renyi divergence can be generalized to the quantum case in many different ways (see e.g. \cite{HiaiMPB2010-f-divergences,LiebFrank2013-Renyi}) due to noncommutativity of quantum states. The most straightforward generalization is the following candidate:
\be
\qrenyisimple_\alpha(\rho||\sigma)=\frac{{\rm sgn}(\alpha)}{\alpha-1} \log\tr \rho^\alpha \sigma^{1-\alpha}.
\ee
For $\alpha=0$ it reduces to the min-relative entropy given by $\qrenyisimple_{\rm min}(\rho||\sigma)=-\log \tr (\Pi_\rho \sigma)$, where 
$\pi_\rho$ is the projector onto the support of $\rho$, and for $\alpha\rightarrow1$ it reduces to the standard quantum relative entropy:
\be
\lim_{\alpha\rightarrow1^+} \qrenyisimple_\alpha(\rho||\sigma)= S(\rho||\sigma)=\tr(\rho \log \rho - \rho \log \sigma).
\label{eq:relent}
\ee
As proved in \cite{marcothesis}, this version of Renyi divergence is known to be monotonic under quantum operations for $\alpha\in[0,2]$,
i.e. for any completely positive map $\Lambda$ and states $\rho$, $\sigma$ we have 
\be
\qrenyisimple_\alpha(\Lambda(\rho)||\Lambda(\sigma))\leq \qrenyisimple(\rho||\sigma),
\ee
for $0\leq\alpha\leq 2$. 

In \cite{Muller-LennertDSFT2013-Renyi,WildeWY2013-strong-converse} (see also  \cite{JaksicOPP2012-entropy})
another version of the quantum divergence was introduced for $\alpha\in(0,\infty]$:
\be
\qrenyi_\alpha(\rho||\sigma)=\frac{1}{\alpha-1}\log 
\left(\tr( \sigma^{\frac{1-\alpha}{2\alpha}} \rho \sigma^{\frac{1-\alpha}{2\alpha}})^\alpha \right) 
\ee
The cases $\alpha=1$ and $\alpha=\infty$ are obtained by limits so that 
$\qrenyi_1(\rho||\sigma)$  is the standard relative entropy \eqref{eq:relent} 
and $\qrenyi_\infty(\rho||\sigma)= \log ||\sigma^{- \frac12}\rho \sigma^{- \frac12}||_\infty$ 
where $||\cdot||_\infty$ is the operator norm. 
Recently it was proven in \cite{LiebFrank2013-Renyi} (see also \cite{Beigi2013-Renyi}) that this entropy is monotonic under quantum completely positive trace preserving maps  for $\alpha\geq 1/2$, namely for any completely positive map $\Lambda$  and states $\rho$, $\sigma$ we have 
\be
\qrenyi_\alpha(\Lambda(\rho)||\Lambda(\sigma))\leq \qrenyi(\rho||\sigma),
\label{eq:qmonotonicity}
\ee
for $\alpha\geq 1/2$. 

Note, that if $\rho$ and $\sigma$ commute, then both types of quantum Renyi divergences reduce to the classical version. 

For both of these quantities, we can further define a Renyi-skew-divergence generalising the skew divergence of \cite{audenaert-skew}. I.e.
\be
\qrenyi_\alpha^s(\rho||\sigma):=\qrenyi_\alpha^s(\rho||\gamma) 
\label{eq:skew}
\ee
with $\gamma=s\rho+(1-s)\sigma$ and $s\in[0,1]$. The definition is similar for $\qrenyisimple$, and they will both turn out to be a monotone under thermal operations when $\sigma$ is the thermal state.

\subsection{Majorization and Schur convexity}
There is a partial order between probability distributions called majorization,
which is defined for arbitrary vectors $x,y \in R^+_k$. We say that $x=(x_1,\ldots x_k)$  majorizes $y=(y_1, \ldots y_k)$ 
if for all $l=1,\ldots k$ 
\be
\sum_{i=1}^l x_i^\downarrow \geq \sum_{i=1}^l y_i^\downarrow , \quad \text{and} \quad  
\sum_{i=1}^k x_i=\sum_{i=1}^k y_i,
\ee 
where $x^\downarrow$ is a vector obtained by arranging the components of $x$ in decreasing order:
$x^\downarrow = (x^\downarrow_1, \ldots, x^\downarrow_k)$ where  $x^\downarrow_1\geq  \ldots \geq  x^\downarrow_k$. We write 
\be
x\succ y
\ee
to indicate that $x$ majorizes $y$.

A function $f$ is called Schur convex if it always preserves the majorization order, 
i.e.  if  $x\succ y $ implies $f(x) \geq f(y)$. If the majorization order is always reversed, the function is called Schur concave. A function is called {\it strictly} Schur convex if $x\succ y $ implies $f(x) > f(y)$ except when $x^\downarrow=y^\downarrow$, then $f(x)=f(y)$. A useful criterion for strict Schur convexity is stated in the following lemma:
\begin{lemma}
A function $f:R^k_+\to R$ of the form $f(x)=\sum_i g(x_i)$  is (strictly) Schur convex/concave,
iff $g(x)$ is (strictly) convex/concave. 
\end{lemma}
This lemma follows from the general criterion of Schur convexity \cite{majorization_Marshall_2011}. However it is easy to prove directly, using the Birkhoff-von Neumann theorem, which states that if $p\succ q$, then $q$ is a convex combination of permutations on $p$.
(cf. Theorem \ref{prop:major}).
Using this, and the strict monotonicity of the logarithm, we see that the R{\'e}nyi entropies reverse the majorization order:
\begin{lemma}
\label{lem:strict_Schur}
For $\alpha\in (-\infty,0) \cup (0,\infty)$, the R{\'e}nyi entropies $H_\alpha$ are strictly Schur concave. 
For $\alpha=0, \pm \infty$, the R{\'e}nyi entropies are Schur concave. The function $\sum_i \log p_i $ 
is also strictly Schur concave.
\end{lemma}


\section{Exact catalysis with trivial Hamiltonian}
\label{sec:trivialHamilExact}
We are interested in the interplay of energy and entropy, which is the essence of thermodynamics.
However, before we approach this problem, it is instructive to first consider the case where the Hamiltonian 
is trivial in the sense that $H=0$, and thermodynamics is reduced to bare information theory. 

This toy model for thermodynamics was described in \cite{uniqueinfo} and it has its roots in 
the problem of exorcising Maxwell demon \cite{Bennett82,Landauer}.
It is constructed within a framework of so called resource theories fo thermodynamics \cite{Beth-thermo,uniqueinfo,HO-limitations,thermoiid} 
having its roots in research on entanglement manipulations \cite{BBPS1996,BBPSSW1996,thermo-ent2002,DevetakHW2005-resource},
where one is interested in transformations between 
states by means of an allowed class of operations. Some states can be brought in for free, 
and they constitute a free resource. The others cannot be created, but only manipulated, i.e. we may transform
one resource state into some other one. 
In the mentioned toy thermodynamics, which is perhaps the simplest known resource theory, the free resource are just maximally mixed states, and all unitary transformations as well as partial trace are allowed operations. 
The emerging class of operations was called {\it noisy operations}. It was shown that in the case of systems of the same size, the class of noisy operations is equivalent to mixtures of unitaries. Therefore the condition that $\rho$ can be transformed into $\rho'$ is equivalent to majorization: $\rho$ can be transformed into $\rho'$ iff the spectrum of $\rho$ majorizes the spectrum of $\rho'$ \cite{Uhlmann-ordering}. The noisy operations are equivalent to thermal operations applied to a system with a trivial Hamiltonian. 


As we mentioned, in this toy thermodynamics, the law that governs state-to-state transitions is 
majorization. However, this is only so when we do not allow ancillary systems that are then returned in the same state. 
It is interesting to analyse how the laws change if we allow catalytic transitions between states. Since the transitions without catalysis 
are governed by majorization, the catalytic transitions will be governed by {\it trumping}. We say that $x$ can be trumped into $y$ if there exists some $z$ such that 
\be
x\ot z \succ y\ot z.
\ee
In \cite{JonathanP} it was for the first time shown that using catalysis, one can perform 
transitions otherwise impossible. 
In particular the following explicit example of  states that do not  majorize one another, but allow catalytic transition 
was given
\be
p=\left(\frac{4}{10},\frac{4}{10},\frac{1}{10},\frac{1}{10}\right), \quad q= \left(\frac12, \frac14, \frac14, 0\right) 
\ee
One checks that $p_1 < q_1$ but  $p_1 +p_2 > q_1 +q_2$. Now, if one takes the catalyst $r=(\frac{6}{10},\frac{4}{10})$,
then $q\ot r \succ p\ot r $, so that $q$ can be trumped into $p$.  

Recently Klimesh and Turgut \cite{Klimesh-trumping,Turgut-trumping-2007} independently provided necessary and sufficient conditions for $x$ to be trumped into $y$.  These were in terms of functions closely related to the Renyi entropies. 
We present here a set of conditions, equivalent to the Klimesh-Turgut ones, written in terms of Renyi divergences w.r.t. the uniform distribution $\eta$. 
We first argue, that original conditions can be equivalently stated in terms of non-strict inequalities,
which in particular allows to remove any discrete subset of conditions by exploiting continuity with respect
to parameter $\alpha$ below. 

\begin{proposition} \label{trumping} 
Let $x\in R_k^+$ and $y\in R_k^+$ be probability vectors which do not both contain components equal to zero. Then $x$ can be trumped into $y$ if, and only if, 
\begin{equation}
D_{\alpha}(x || \eta) \geq D_{\alpha}(y || \eta),
\end{equation}
for all $\alpha\in (-\infty, \infty)$, with $\eta = (1/k, \ldots, 1/k)$ the uniform distribution. 
\end{proposition}

\begin{proof}
Klimesh \cite{Klimesh-trumping} proved that if $x$ and $y$ do not both contain components equal to zero, 
than $x$ can be trumped into $y$ if, and only if, for all $\alpha\in (-\infty,\infty)$:
\be \label{strict}
f_\alpha(x)> f_\alpha(y),
\ee
where 
\beq
f_\alpha(x)=\left \{\bea{ll}
\log\sum_{i=1}^k x_i^\alpha &(\alpha>1); \\
\sum_{i=1}^k x_i \log x_i & (\alpha = 1);\\
-\log \sum_{i=1}^k x^\alpha_i &(0 < \alpha < 1);\\
-\sum_{i=1}^k \log x_i & (\alpha=0);\\
\log \sum_{i=1}^k x^\alpha_i & (\alpha < 0),
\eea\right.
\label{eq:f}
\eeq

Let us first show that we can replace strict inequalities with non-strict ones. Obviously, when $x$ can be trumped 
into $y$ then the conditions with non-strict inequalities are satisfied. 

Suppose, conversely, that conditions \eqref{strict} are satisfied for all $\alpha\in R$ but with non-strict inequalities.
Then consider $y_\epsilon=(1-\epsilon) y +\epsilon \eta$ where $\eta$ is the uniform distribution. We have that $y^\downarrow\not=y^\downarrow$, and  $y^\downarrow \succ y_\epsilon^\downarrow$. 
All the functions $f_\alpha$ are strictly Schur convex, so we have $f_\alpha(x)>f_\alpha (y_\epsilon)$. 
Therefore, $x$ can be trumped into $y$, i.e. there exists $z$ such that   $x\ot z \succ y_\epsilon \ot z$. 
since this is true for arbitrary $\epsilon$, we must have also that $x\ot z \succ y \ot z$, 
as majorization conditions are continuous - if the perturbation is small enough, it does not change the ordering, 
and the conditions are inequalities for continuous functions of the ordered probability distirbutions.

Now, all the functions $f_\alpha$, excluding $\alpha=0$ are proportional to $-H_\alpha(x)$. In turn, the R{\'e}nyi entropy is related to $D_\alpha (x\|\eta)$ by  Equation \eqref{interchangeHD}. 
The conditions  $\alpha=0$  we recover as follows:
first we note that if $x$ or $y$ do not have full rank, then $f(x)$ or $f(y)$ diverges to infinity, and this feature is already detected when $\alpha<0$. If both $x$ and $y$ have full rank, then it can be verified that
\begin{equation}
-\frac{1}{k} \sum_{i=1}^k \log x_i - \log k = \lim_{\alpha\rightarrow 0^+} \frac{1-\alpha}{\alpha} D_\alpha (x\|\eta).
\end{equation}
\end{proof}

\begin{remark} Knowing that non-strict inequalities are enough,  one can remove any discrete amount of trumping conditions, as they can be obtained in the limit from the other conditions.
\end{remark}


The Klimesh-Turgut conditions have a peculiar feature. Namely, if there are zeros in both $x$ and $y$ 
we have to truncate the zeros and compute the conditions on smaller vectors.
This means that we do not simply compare two functions, but which function we will choose depends on 
some relation between the vectors. Second, the condition of having zeros in both vectors is very unstable. 
If we slightly perturb $y$ so that we remove zeros, then we do not need to truncate anymore. Note that the problem here 
is only with negative $\alpha$, as for positive $\alpha$ the functions $D_\alpha( p || \eta)$ do not depend on additional zeros, 
while the functions with negative $\alpha$ are infinite, when at least one component vanishes.

First let us note, that we can eliminate the first problem by requiring that the output is returned not exactly,
but only with arbitrary high accuracy. Namely we have 


\begin{proposition} \label{trumping-better} 
Let $x\in R_k^+$ and $y\in R_k^+$ be probability vectors. Then the following conditions are equivalent:
\begin{enumerate}
\item For arbitrary $\epsilon>0$, there exists  $y_\epsilon$ such that $||y-y_\epsilon||\leq \epsilon$ and $x$ can be trumped into $y_\epsilon$. 
\item  The following inequality:
\begin{equation} 
D_{\alpha}(x || \eta) \geq D_{\alpha}(y || \eta)
\label{eq:DxDy2}
\end{equation}
holds for all $\alpha\in (-\infty, \infty)$, with $\eta = (1/k, \ldots, 1/k)$ being the uniform distribution. 
\end{enumerate}
\end{proposition}

\begin{proof}
Let us first note that without loss of generality we can assume that $y_\epsilon$ is of full rank.
Indeed, if it is not,  we can perturb the state to add a small amount 
of noise (which is a valid noisy operation), obtaining $\tilde y_{\epsilon'}$ which has full rank, i.e. no zeros,
and stisfies $||\tilde y_{\epsilon'}-y||\leq\epsilon'$, with $\epsilon'$ arbitrarily small.

"Condition 1 implies 2": Assume that for any $\epsilon$ we can trump $x$ into $y_\epsilon$. 
We can then apply proposition \ref{trumping} to obtain 
\begin{equation}
D_{\alpha}(x || \eta) \geq D_{\alpha}(y_{\epsilon} || \eta).
\label{eq:DxDy3}
\end{equation}
Now, we use continuity of $D_\alpha$ in first argument for fixed full rank second argument,
and by letting $\epsilon\to 0$ obtain \eqref{eq:DxDy2}.

"Condition 2 implies 1:" Assume that \eqref{eq:DxDy2} hold.
Consider then $y_\epsilon=(1-\epsilon) y + \epsilon \eta$ where $\epsilon\in (0,1]$ can be arbitrarily small. We then have 
\be
D_{\alpha}(x || \eta) > D_{\alpha}(y_{\epsilon} || \eta)
\ee
since admixing $\eta$ can only decrease $D_\alpha (y || \eta)$. Applying proposition \ref{trumping} tells us that $x$ can be trumped into $y_\epsilon$. 
\end{proof}

We note that earlier, Aubrun and Nechita \cite{AubrunNechita-catalysis} gave conditions for trumping, in which only $H_\alpha$ with $\alpha> 1$ were needed. This is because they considered a special kind of closure,  where one is allowed to add an arbitrary number of zeros to the initial vector $x$ while returning an (arbitrarily good) approximation of the needed output $y$. This is a kind of embezzling (see section \ref{subsec-embezzling} for definition of embezzling): 
One adds an ancilla and returns it with arbitrary small error, but the size of ancilla needs to grow in order to make the error smaller. In thermodynamics this is not allowed, as according to the second law, we should 
consider processes which do not change
the environment. The proposition \ref{trumping} already suggests how the conditions for thermodynamics should look like: the maximally mixed state should get replaced by the Gibbs state. In section \ref{sec:adultHamil} we will prove that this is indeed the case.

\subsection{Restricting to $\alpha \geq 0$ by investing a small amount of extra work}
\label{sub:extra_work}
In this section, we will argue  that if we are allowed to invest an arbitrary small amount of work,
only the conditions with positive $\alpha$ 
are relevant. This was considered by Aubrun and Nechita in \cite{AubrunNechita-catalysis}, and the way we do that below is very similar.

We have three cases
\begin{itemize}
\item[(i)] $p$ has more zeros than $q$.
\item[(ii)] $p$ has less zeros than $q$.
\item[(iii)] $p$ and $q$ have the same number of zeros (in particular, can both have no zeros).
\end{itemize} 
In case $(i)$, after truncation $p$ will still have zeros, and so relative entropies for $p$ with negative $\alpha$ will be infinite. Therefore the conditions with negative $\alpha$ are always satisfied. In case $(ii)$ the transition cannot be realised, but this is reported by comparing ranks, 
which can be obtained using $H_0=\lim_{\alpha\to 0^+}H_\alpha$. So again already the conditions with $\alpha>0$ report the impossibility of transition, and negative $\alpha$'s are not needed. Finally, in case $(iii)$, if we consider the  transition $p \to q$, with a small amount of work invested in addition:
\be
p \ot|0\>_{d+1}\<0| \ot \eta_{d}  \to q \ot |0\>_{d}\<0| \ot \eta_{d+1},
\ee
where $|0\>_{k}\<0|$ stands for the distribution $\underbrace{(1,0, \ldots, 0)}_k$ (equivalently, a pure state  on $C^k$) and $\eta_k$ stands for the uniform distribution with $k$ elements. The invested amount of work is  $\log \frac{d+1}{d}$, hence it is arbitrary small, when $d\to \infty$.
Now, on the left hand side, we have more zeros than on the other side (as initially, we had the same number of zeros). Therefore, we are back to case $(i)$. 

Note however, that we can rule out negative $\alpha$ in a different way, if instead of insisting on preparing exact output state, we allow for preparing an $\varepsilon$-approximation, with arbitrary accuracy (as discussed in Prop. \ref{trumping-better}). Therefore, in case $(iii)$ considered above we can add to both sides an ancilla in a pure state of the same dimension (which is actually needed only 
when $p$ and $q$ both have no zeros). Then both input and ouptut will have the same number of zeros. 
But as the desired output we can take an approximation of
it which is full rank. Then the transition is governed solely by the conditions with positive $\alpha$. The returned state of ancilla is now only approximately pure, but the accuracy is arbitrarily good. Note that we can choose the approximation in such a way that it affects only the ancilla; the original output state 
is not changed and will be produced exactly. 

The above two methods of ruling out negative $\alpha$ do not differ very much: in the first one we input some pure state 
of large dimension and return a pure state of dimension smaller by $1$, while in the second we input a qubit in a pure state 
and return it with arbitrarily good approximation. However, in both cases one has to borrow a perfectly pure state, which in 
thermodynamic context would mean initially investing an infinite amount of work.

\section{Cathalytic thermal operations}
\label{sec:TO}
In Sec. \ref{sec:trivialHamilExact}, we considered catalytic noisy operations, where apart from using noisy operations, one can add any catalyst, provided it is returned in the same state. 
To obtain a full picture of thermodynamic interactions, we consider a generalisation of these operations to the setting where we have non-trivial Hamiltonians, and a class of {\it thermal operations} introduced in \cite{Beth-thermo,HO-limitations,thermoiid}. The relevant objects now come in pairs $(\rho,H)$ - the system's state and its Hamiltonian. 

Consider a fixed temperature $T$. Then given a system $S$ in some initial state $\initial$, thermal operations include the following procedures:
\bei
\item adding any system $\gibbs^R$ which is the thermal state w.r.t. an arbitrary Hamiltonian $H_R$ at temperature $T$,
\item applying arbitrary unitary that commutes with the joint Hamiltonian of system and reservoir $H_S+H_R$,
\item removing a system by performing partial trace.
\eei  
It is clear that thermal operations preserve the thermal state. Conversely, for states diagonal in energy basis, 
any operation which preserves the thermal state belongs to the class of thermal operations as proved in 
\cite{Beth-thermo,HO-limitations}. 

We now consider catalytic thermal operations, where in addition to thermal operations, 
one can add an ancilla called the catalyst $(\rho_C, H_C)$, provided it is returned in the same state (in product form with the other systems). In other words, $(\rho, H)$ can be catalytically transformed into $(\rho',H')$ 
if there exists catalyst $(\rho_C, H_C)$ such that $(\rho\ot\rho_C, H+H_C)$ 
can be transformed into $(\rho' \ot \rho_C, H'+H_C)$. Demanding that the catalyst be returned uncorrelated from the system is natural, since one is likely to want to re-use the catalyst over many cycles, and having correlations between the catalyst and the system could result in inefficiencies later on. One could also imagine that the purification of the catalyst is held somewhere, and require that the catalyst and its purification remain close to its original state. 

\section{Equivalence relations from set of allowed resources defines the zeroeth law}
\label{sec:zeroeth}
In this section, we spell the proof of the zeroeth law in this framework of thermal operations, which singles out Gibbs states as the unique free resource in thermodynamics. This has been shown to be true \cite{thermoiid} for qubits, however we provide a proof here for general systems. More precisely, if we consider the resource theory framework of thermal operations, which allows arbitrary energy preserving unitaries across the global system, then the only resource states that, when allowed for free, does not give rise to trivial state conversion conditions is that of the Gibbs state. It is clear that this is directly related to work extraction via thermal operations, since one should make sure that work cannot be extracted from a freely allowed resource. Otherwise, one could, by using indefinite amounts of free resources, extract enough work to facilitate any state transformation. Therefore, we want to show that given a system Hamiltonian $H_A$, and the corresponding thermal state $\tau_A$ with some fixed temperature $T$, then by allowing arbitrarily many copies of any other state $\rho_A \neq \tau_A$, one can always extract work deterministically by thermal operations. 

Let us first show that this is true for states diagonal in the energy eigenbasis. The case where coherences between energy eigenstates exist can be dealt with, as shown in \cite{Skrzypczyk13}, by decohering multiple copies of the state in its energy subspace, and invoking Landauer's principle to extract work. We shall base on the result of Pusz and Woronowicz \cite{pusz_passive_1978} who introduced the notion of {\it passive} states,  i.e. the states whose energy cannot be decreased by arbitrary cyclic Hamiltonian evolution. 
They further introduced a notion of {\it completely passive} states:

\begin{definition}[Completely passive states]
A state $\rho$ is completely passive iff for all $n\in\mathbb{Z}^+$, $\rho^{\otimes n}$ is also passive.
\end{definition}

They proved for general quantum systems described by a $C^*$ algebra that the only completely passive states are either so called KMS states or ground states. Lenard \cite{lenard78} has translated their results to the case of  finite-dimensional systems. 
For such systems one has:
 

\begin{fact}
Consider a state $\rho$ corresponding to some Hamiltonian $H$, and let $\lbrace p_i\rbrace$ and $\lbrace E_i\rbrace$ be the eigenvalues of $\rho$ and $H$ respectively. Then $\rho$ is passive, iff $[\rho,H]=0$, and for any $i,j$, $E_i > E_j$ implies that $p_i \leq p_j$.
\end{fact}

Moreover a state which is completely passive is either ground state or a Gibbs state (since in finite dimensional state,
the KMS state is unique and it is just the Gibbs state). 

Either from the very definition or from the above characterization one sees that given any non-passive state $\rho_{np}$, one can extract on \textit{average} a non-zero amount of work by performing a population inversion (a switch of energy levels). By invoking typicality arguments, we show that by performing extraction over many identical copies of $\rho_{np}$, an amount of work closely related to average work can always be extracted except with small probability. Then to conclude, one uses the mentioned result, that any state which is not a Gibbs state 
or a ground state, becomes non-passive, once we take sufficiently many copies. 

Before beginning the proof, we state for completeness Theorem \ref{th:hoeffding} which is the main tool we use to invoke typicality.
\begin{theorem}[Hoeffding inequality]\label{th:hoeffding}
Consider $x = \sum_{i=1}^N h_i$, where $h_1,\cdots,h_N$ are independent random variables such that for every $i, h_i\in[a_i,b_i]$. Denote $R_i=b_i-a_i$. Then for any $\alpha>0$, the probability
\begin{equation}
P\left[|x-E[x]| \geq \alpha N\right]\leq e^{-\frac{2\alpha^2N^2}{\sum_i R_i^2}}.
\end{equation}
\end{theorem}

\begin{theorem}\label{th:workext}
Given any non-passive, diagonal state $\rho$ corresponding to some Hamiltonian $H$. Then for any probability $\varepsilon>0$, there exists $m$ such that given $m$ copies of $\rho$, it is possible to extract a non-zero amount of work, except with probability $\varepsilon$.
\end{theorem}

\begin{proof}
There are two main steps in this proof: first we construct the unitary that performs work extraction from a single copy of $\rho$, while storing the work in a battery system $B$ similar to that of \cite{Skrzypczyk13}. Subsequently, we form the joint unitary over the $m$ systems and the battery system, and use the Hoeffding inequality to bound the amount of extracted work, albeit with some small failure probability.

Consider a non-passive $\rho_{A_k}$ on system $A_k$, $1\leq k\leq m$. Then by the definition of passivity, there exists some $i,j$ where $E_i > E_j$, $p_i > p_j$ holds. Construct the battery as a harmonic osciilator system with Hamiltonian $H_B = \sum_{n=1}^N n\hbar\omega |n\rangle\langle n|$, where $\hbar\omega=E_i-E_j$. Define the lowering operator $\hat{a}_B=\sum_{i=1}^N |i-1\rangle\langle i|$ and the joint energy-preserving unitary over system $A_k$ and $B$,
\begin{equation}
U_k = |i\rangle\langle j|_{A_k}\otimes \hat{a}_B + |j\rangle\langle i|_{A_k} \otimes \hat{a}^\dagger_B + |i\rangle\langle i|_{A_k}\otimes|N\rangle\langle N|_B+|j\rangle\langle j|_{A_k}\otimes|0\rangle\langle 0|_B+\sum_{r\neq i,j} |r\rangle\langle r|_{A_k} \otimes \mathbb{I}_B.
\end{equation}
Clearly, this unitary acted upon the initial joint state $\rho_{A_k}\otimes|m\rangle\langle m|_B$ extracts $\hbar\omega \geq 0$ amount of work into the battery system with probability $p_i$, and $-\hbar\omega$ with probability $p_j$, while the expectation value of work extracted is given by the difference in expected energy in the battery system, $\langle W \rangle = \hbar\omega (p_i-p_j) > 0$.

Now, similarly, consider the initial state $\rho_{A_1}\otimes \cdots \otimes \rho_{A_m} \otimes |m\rangle\langle m|_B$, and the unitary transformation $U_m \cdots U_1$. Each unitary raises, lowers or leave unchanged the battery state with certain probabilities; hence this operation can be represented with a string $x = x_1\cdots x_m$ numbers, where for all $y=1, \cdots, m$, 
\begin{equation}\label{eq:probdist}
p(x_y=c)=\left \{
\begin{array}{ll}
p_i & (c=1); \\
p_j & (c=-1);\\
1-p_i-p_j &(c=0).
\end{array}
\right.
\end{equation}

The total amount of work extracted in this process is equal to $W_T=x_t \hbar\omega$, where $x_T=\sum_{i=1}^m x_i$, while the expectation value of $x_T$ is $\langle x_T\rangle = m (p_i-p_j)>0$. Since $x_1, \cdots, x_m$ are i.i.d. random variables with bounded values between -1 and 1, we can invoke Theorem \ref{th:hoeffding} for some small number $\alpha>0$, obtaining
\begin{equation}
P[x_T \leq \langle x_T\rangle-\alpha m] \leq P[|x_T-\langle x_T\rangle |\geq \alpha m] \leq {\rm exp}\left[\frac{-2\alpha^2m^2}{\sum_{i=1}^m 2^2}\right] = {\rm exp}\left[\frac{-\alpha^2m}{2}\right] = \varepsilon,
\end{equation}
Therefore, we conclude that $W_T \geq \left(\langle x_T\rangle-\alpha m\right) \hbar\omega >0$ except with probability $\varepsilon$. For any $\varepsilon>0$, pick some small $\alpha$ and $m=\frac{2}{\alpha^2}\log\frac{1}{\varepsilon}$ suffices.
\end{proof}

\begin{theorem}\label{th:zeroeth}
Given a fixed Hamiltonian $H$, if $\rho'$ is not the Gibbs state or the ground state corresponding to the Hamiltonian $H$ at a certain temperature $T$, then for any $\varepsilon>0$, there exists $n$ such that given $n$ copies $\rho'$, it is possible to extract a non-zero amount of work, except with probability at most $\varepsilon$.
\end{theorem}
\begin{proof}
The proof comes from a straightforward realisation that any such $\rho'$ is not \emph{completely} passive, hence there exists some positive integer $a$ such that $\rho'^{\otimes a}$ is non-passive. One can then invoke Theorem \ref{th:workext} with $\rho = \rho'^{\otimes a}$. For any $\varepsilon>0$, pick some small $\alpha$ and $n=am=a\frac{2}{\alpha^2}\log\frac{1}{\varepsilon}$ suffices.
\end{proof}

\section{Catalysis with general Hamiltonians}
\label{sec:adultHamil}

We now turn to the case of the full theory of thermodynamics, where we have an interplay between energy and entropy. In this case, the Hamiltonian of the system and reservoir are arbitrary.
To derive the conditions for state transformations, we will need a generalisation of the majorisation condition, known as $d$-majorisation, which we describe in \ref{ss:dmaj}. We then derive the catalytic version of
this, in \ref{sec:ruch_trumping}, theorem \ref{th:catalysis}. Next we show that it is sufficient to consider catalysts diagonal w.r.t its Hamiltonian, when the initial state is already block diagonal in the energy eigenbasis.
We can then apply our results to the case of thermodynamics to derive second laws -- these are stated seperately in Section \ref{sec:secondlaws}. We will later in Section \ref{sec:inexact} discuss the 
case where the thermodynamic processes do not return the catalyst in exactly the same state, but only approximately close.

\subsection{$d$-majorization} 
\label{ss:dmaj}

\begin{definition}
Given probability distributions $p, q, p', q'$, we say that $(p,q)$ $d$-majorizes $(p',q')$ if and only if for any convex function $g$, 
\be
\sum_i q_i g\left(\frac{p_i}{q_i} \right) \geq \sum_i q_i' g \left(\frac{p_i'}{q_i'} \right).
\ee
We denote this as $d(p || q) \succ d(p' || q')$.
\end{definition}

The Birkhoff-von Neumann theorem relates majorization to transition between states, which we state here:
\begin{theorem}
For two probability distributions $p$ and $p'$ the following two conditions are equivalent:
\bei
\item[(i)] $p$ majorizes $p'$.
\item[(ii)] there exists a channel $\Lambda$ such that 
\be
\Lambda(p) = p', \quad \Lambda(\eta)=\eta,
\ee
where $\eta$ is the uniform distribution.
\eei
\label{prop:major}
\end{theorem}

The generalization of Birkhoff-von-Neumann theorem to $d$-majorization is as follows:

\begin{proposition}[Theorem 2, \cite{RuchSS-1978}]
For probability distributions $p, p', q, q'$ the following two conditions are equivalent:
\bei
\item[(i)] $(p,q)$ $d$-majorizes $(p',q')$,
\be
d(p || q) \succ d(p' || q').
\ee
\item[(i)] There exists a channel $\Lambda$ such that 
\be
\Lambda(p)=p',\quad \Lambda(q)=q'.
\ee
\eei
\label{prop:Ruch}
\end{proposition}

In fact it has been showned that a particular limited set of convex functions is sufficient. In the case when $q=q'$, the conditions can be  
expressed by the so-called {\it thermo-majorization} diagrams 
identified in \cite{HO-limitations}. The thermo-majorization diagrams can be easily extended also to the case $q \neq q'$, although this is not relevant for our application to thermodynamics in which $q=q'$ is the Gibbs state. 

One can see that Proposition \ref{prop:Ruch} implies Theorem \ref{prop:major}, by taking $q=q'=\eta$ and noting that
\be
p \succ p' \quad \leftrightarrow \quad d(p |\eta) \succ d(p'|\eta).
\ee

A natural question is whether there is a trumping analogue for $d$-majorization. We show that this is the case
in section \ref{sec:ruch_trumping}, where we prove a version of Prop. \ref{prop:Ruch} allowing catalysis. We will recover analogous relations to the Klimesh-Turgut ones in the case where $q=q'=\eta$. We illustrate the situation in Table \ref{tab:orders}.
\begin{table*}[h!]
	\centering
		\begin{tabular}{|c|c|c|} \hline
			& standard majorization & d-majorization \\
			\hline
		no catalysis	& $d(p|\eta) \succ d(p|\eta)$ & $d(p|q) \succ d(p'|q')$ \\ \hline
		            	& $\exists \Lambda: \,\,  \Lambda(p)=p',\quad  \Lambda(\eta)=\eta$ &
		            	$\exists \Lambda: \,\,  \Lambda(p)=p', \quad  \Lambda(q)=q'$ \\ \hline
		            	& $\forall_l\sum_{i=1}^l p_i\geq \sum_{i=1}^l p_i'$ & comparing diagrams (Fig. \ref{fig:thermomaj})
		            	 \\ \hline \hline
	    catalysis	& $D_\alpha(p\|\eta) \geq D_\alpha(p'\|\eta)$ 
	     & $D_\alpha(p\|q) \geq D_\alpha(p'\|q')$ \\ \hline
	    & $\exists \Lambda,r: \,\,  \Lambda(p\ot r )=p'\ot r, \quad  \Lambda(\eta\ot \tilde \eta)=\eta\ot \tilde \eta$ &
		            	$\exists \Lambda,r,s: \,\,  \Lambda(p\ot r )=p'_\ep\ot r,\quad  \Lambda(q\ot s)=q'\ot s $ \\ \hline
		\end{tabular}
	\caption{Partial orderings as criteria for state transformations. $p'_\epsilon$ is an arbitrarily good approximation of $p'$.}
	\label{tab:orders}
\end{table*}

\subsection{Notations and technical tools} 
\label{ss:dmaj_catalysis_tools}

Before setting out to prove the main result in Section \ref{sec:ruch_trumping}, we list several technical tools and lemmas that we will use.

We will start by describing an embedding channel, which will be used later to prove the above results. Consider the simplex of probability distributions $\{p_i\}_{i=1}^k$ and a set of natural numbers $\{d_i\}_{i=1}^k, i=1,\ldots, k$, and let $N=\sum_{i=1}^k d_i$. 
We define the embedding 
\be 
\Gamma(p)= \oplus_{i} p_i \eta_i ,
\label{eq:gamma}
\ee
where $\eta_i = \{\frac{1}{d_i} ,\ldots, \frac{1}{d_i}\}$ is the uniform distribution. More clearly, the image can be writen as the following $N$-dimensional probability distribution:
\be
\Gamma(p) = \biggl\{\underbrace{\frac{p_1}{d_1},\ldots, \frac{p_1}{d_1} }_{d_1},\ldots, 
\underbrace{\frac{p_k}{d_k},\ldots, \frac{p_k}{d_k} }_{d_k} \biggr\}.
\ee 
The inverse map $\Gamma^*$ acts on the space of $N$-dimensional probability distributions $\tilde p$, where
\be
\tilde p= \oplus_{i=1}^k \tilde p^{(i)},
\ee
with each $\tilde p^{(i)} = \{ \tilde p_1^{(i)}, \cdots, p_{d_i}^{(i)} \}$ being a (unnormalized) $d_i$-dimensional probability distribution. $\Gamma^*$ can be written as 
\be
\Gamma^*(\tilde p) = r,
\ee
with $r=\{r_i\}_{i=1}^k$ being a normalized probability distribution with the form $r_i=\sum_{j=1}^{d_i}\tilde p_j^{(i)}$.
The maps $\Gamma$ and $\Gamma^*$ are channels, and for all probability distributions $p$ we have $\Gamma^* (\Gamma(p)) =p$. Moreover, for the specific state $\gamma = \{ \frac{d_1}{N}, \ldots, \frac{d_k}{N} \}$, 
\be
\label{embeddedgamma}
\Gamma(\gamma) = \id/N
\ee 
is the uniform $N$-dimensional distribution. Recall that $\sum_{i=1}^k \frac{d_i}{N} = N\cdot \frac{1}{N} =1$, and $\gamma$ is normalized.

Also, we will be making use of the following simple lemmas, which we state here for the reader's convenience. Lemma \ref{lem:rel_ent_H} relates the R{\'e}nyi divergence $D_\alpha (p\|q)$ to $D_\alpha (\Gamma(p)\|\Gamma(q))$. Lemma \ref{strictaux} describes the change in $D_\alpha (p\|q)$ when $p$ is replaced by a mixture of distributions $p$ and $q$. Lemma \ref{lem:channelrevised} is a technical tool that enables us to work with distributions containing irrational probability values, by introducing small corrections such that we need only to consider rational values. Lastly, Lemma \ref{lem:direct_sum} shows us when a channel can be writen as a direct sum of two channels acting disjointly on partitions of the total input/output space. 

\begin{lemma}
\label{lem:rel_ent_H}
Let $p=\{p_i\}_{i=1}^k$ be an ordered probability distribution, and $\{d_i\}_{i=1}^k$be natural numbers with $\sum_{i=1}^k d_i=N$.
Define the following fine-grained, $N$-dimensional probability distribution
\be
\tilde p = \biggl\{\underbrace{\frac{p_1}{d_1},\ldots, \frac{p_1}{d_1} }_{d_1},\ldots, 
\underbrace{\frac{p_k}{d_k},\ldots, \frac{p_k}{d_k} }_{d_k} \biggr\},
\ee 
and let $\gamma$ be the $k$-dimensional probability distribution 
\be
\gamma = \left \{ \frac{d_1}{N}, \ldots, \frac{d_k}{N} \right \}.
\ee
Then for $\alpha\in[-\infty, \infty]$ we have 
\be
D_\alpha(p ||  \gamma) =  D_{\alpha}( \tilde p || \eta_N),
\ee
with $\eta_N = (1/N, \ldots, 1/N)$ the uniform distribution. By \eqref{embeddedgamma}, this means that equivalently 
\be
D_\alpha(p ||  \gamma) =  D_{\alpha}( \Gamma (p) || \Gamma(\gamma)).
\ee
\end{lemma}

\begin{proof}
Let us first assume that $\alpha\not \in \{-\infty,0,1,\infty\}$. Then
\begin{align}
D_\alpha(p || \gamma) & = \frac{1}{\alpha-1} \log\sum_{i=1}^k p_i^\alpha \left(\frac{d_i}{N} \right)^{1-\alpha}\\
& = \frac{1}{\alpha-1} \log\sum_{i=1}^k d_i \left(\frac{p_i}{d_i}\right)^\alpha \left(\frac{1}{N} \right)^{1-\alpha}\\ 
&= \frac{1}{\alpha-1} \log\sum_{i=j}^N \tilde{p}_j^\alpha \left(\frac{1}{N} \right)^{1-\alpha}=D_{\alpha}(\tilde p || \eta_N),
\end{align}
where note that in the third inequality we sum over the fine-grained distribution $\tilde{p}$, and for each $i\in\{1,\cdots,k\}$, $\tilde{p}$ contains $d_i$ number of degenerate values $\frac{p_i}{d_i}$.
For $\alpha\in\{-\infty, 0,1,\infty\}$ one can obtain the relation by considering limits. Here we show this explicitly. For $\alpha=0$ we have 
\be
D_0( p ||\gamma)= -\log \sum_{i:p_i\not=0} \gamma_i=-\log \sum_{j:p_j\not=0}\frac{d_i}{N}=
-\log \sum_{i:\tilde{p}_i\not=0} \frac{1}{N} = D_{0}( \tilde{p} || \eta_N).
\ee
For $\alpha\rightarrow\infty $, by using \eqref{infdivergence} we have 
\be
D_\infty(p \| \gamma) =\log \min \left\{\lambda: \forall i, \lambda \geq \frac{p_i}{d_i/N}\right\} =
\log \min \left\{\lambda:\forall i, \lambda\geq\frac{p_i/d_i}{1/N}\right\} = D_{\infty}(\tilde p || \eta_N),
\ee
and similarly for $\alpha\rightarrow-\infty$. Finally, for $\alpha\rightarrow1$, 
\be
D_1( p || \gamma )=\sum_{i=1}^k p_i \log \frac{p_i N}{d_i} = \sum_{i=1}^k d_i \frac{p_i}{d_i} \log \frac{p_i N}{d_i} = D_1(\tilde p || \eta_N).
\ee
\end{proof}

\begin{lemma} \label{strictaux}
Let $p \neq q$ be distributions with $q$ full rank. Then for all $\alpha \in (- \infty, \infty)$ and all $0 < \delta < 1$,
\begin{equation}
\label{strictineq}
D_{\alpha}((1 - \delta)p + \delta q || q) < D_{\alpha}(p  || q).
\end{equation}
\end{lemma}

\begin{proof}
The proof is given separately for different ranges of $\alpha$.\\
In the range of $\alpha\in[0,1]$, \eqref{strictineq} is obtained by the joint convexity of the R{\'e}nyi divergence, as proven in Theorem 11 of \cite{timerven2012}. The special case of $\alpha=1$ is also proven in \cite{thomascover}. For any two sets of distributions $(A_1,B_1)$ and $(A_2,B_2)$, and a parameter $0<\delta<1$, the R{\'e}nyi divergence in this range of $\alpha$ satisfies 
\begin{equation}
D_\alpha ((1-\delta)A_1+\delta A_2 \| (1-\delta)B_1 + \delta B_2) \leq (1-\delta)D_\alpha (A_1\|B_1) + \delta D_\alpha (A_2\|B_2).
\end{equation}
Setting $B_1 = B_2 = q$, $A_1 = p$ and $A_2 = q$, then 
\begin{equation}
D_\alpha ((1-\delta)p+\delta q \| q) \leq (1-\delta)D_\alpha (p\|q) + \delta D_\alpha(q\|q) < D_\alpha (p\|q)
\end{equation}
since $D(q\|q)=0$, $\delta>0$ and $D_\alpha (p\|q)>0$ for $p\neq q$.

For $\alpha>1$ joint convexity does not hold. However, denote $r=(1-\delta)p+\delta q$, note that \eqref{strictineq} is equivalent to
\begin{equation}\label{convexsum}
\sum_i r_i^{\alpha} q_i^{1-\alpha} > \sum_{i=1}^k p_i^\alpha q_i^{1-\alpha}.
\end{equation}
We now make use of the fact that the function $f(x)=x^\alpha$ is convex for $\alpha>1$, over all $x>0$ (This be checked since the second derivative of $f(x)$ is positive over the range of $x>0$). and since each $q_i^{1-\alpha}>0$, the linear combination 
\begin{equation}
 F(p) = \sum_i q_i^{1-\alpha} p_i^\alpha
 \end{equation} 
is convex w.r.t. $p$. For any $p\neq q$, note that $F(p)> 1$ comes from the positivity of the R{\'e}nyi divergences, since $D_\alpha(p\|q) = \frac{\sgn(\alpha)}{\alpha-1} \log F(p)$. Also, note that for a fixed $q$, $F(q) = \sum_i q_i = 1$.
This implies that for \eqref{convexsum},
\begin{align}
{\rm L.H.S.} \equiv F(r) &\leq (1-\delta) F(p) + \delta F(q)\\
&= F(p) -\delta [F(p)-F(q)]\\
&< F(p) \equiv {\rm R.H.S.}
\end{align}
For $\alpha<0$, whenever the kernel of $p$ is non-empty, $D_\alpha (p\|q)=\infty$, However $D_\alpha ((1-\delta)p+\delta q\|q)$ is finite since $q$ has full rank, hence the desired inequality holds always. When $p$ also has full rank, the same approach can be used as in $\alpha>1$, since $f(x)=x^\alpha$ is also convex for positive $x$, and since $\frac{\sgn(\alpha)}{\alpha-1}>0$, \eqref{strictineq} is also equivalent to \eqref{convexsum}.
\end{proof}

\begin{lemma}
\label{lem:channelrevised}
Let $q=\{q_i\}_{i=1}^k$ be an ordered (descendingly) $k$-dimensional probability distribution of full rank, possibly containing irrational values. Then for any $\varepsilon>0$, there exists a state $\tilde{q}$ such that
\bei
\item[(i)]  $\| q-\tilde{q}\|\leq \varepsilon$,
\item[(ii)] Each entry of $\tilde{q}$ is a rational number, i.e. there exists a set of natural numbers $\{d_i\}_{i=1}^k$ such that $\tilde{q}= \{\frac{d_i}{N} \}_{i=1}^k$, with $\sum_{i=1}^k d_i = N$,
\item[(iii)] There exists a valid channel $E$ such that $E(q)=\tilde{q}$ and for any other distribution $p$, $\|p-E(p)\|\leq O(\sqrt{\varepsilon})$.
\eei
\end{lemma}

\begin{proof}
We prove this by constructing a specific $\tilde{q}$ that satisfies (i) and (ii), then we construct the corresponding channel $E$ that maps $q$ to $\tilde{q}$. Note that since $q$ is ordered, the minimum value $\min_i q_i = q_k$. Firstly, choose any integer $N$ such that $q_k > \frac{k}{\sqrt{N}} > \frac{k}{N}$. Now, for $i\in\{1,\cdots,k-1\}$, define
\begin{equation}
\tilde q_i = \frac{\ceil{q_i N}}{N} = \frac{d_i}{N}.
\end{equation}
Lastly, choose 
\be
\tilde q_k = 1-\sum_{i=1}^{k-1} \tilde q_i = 1-\frac{\sum_{i=1}^{k-1} d_i}{N} = \frac{d_k}{N}.
\ee

It is clear that 
\be
1-\tilde q_k = \sum_{i=1}^{k-1} \tilde q_i & = \sum_{i=1}^{k-1} \frac{\ceil{q_i N}}{N}\\
& \leq \sum_{i=1}^{k-1} \frac{q_i N}{N} + \frac{k-1}{N}\\
& \leq 1- q_k + \frac{k}{N},
\ee
which means that $\tilde q_k \geq q_k - \frac{k}{N} >0$ from our choice of $N$, hence $\tilde q$ is a valid distribution with rational probability values, and condition (ii) is satisfied. Note also that the statistical distance between $q$ and $\tilde{q}$ is lower bounded by
\be
\Vert q - \tilde q \Vert = \frac{1}{2} \sum_{i=1}^k |q_i -\tilde{q}_i| &= \frac{1}{2} \left[\sum_{i:q_i > \tilde q_i} (q_i -\tilde{q}_i) + \sum_{i:q_i \leq \tilde q_i} (\tilde{q}_i - q_i)\right]  \\
& = \sum_{i:q_i > \tilde q_i} q_i - \tilde q_i \\
& = q_k - \tilde q_k \leq \frac{k}{N} = \varepsilon
\ee
which satisfies condition (i) by setting $\varepsilon=\frac{k}{N}$, which can be arbitrarily small by choice of $N$.
The third inequality holds because both distributions are normalized,
\be
\sum_{i:q_i > \tilde q_i} q_i + \sum_{i:q_i \leq \tilde q_i} q_i &= \sum_{i:q_i > \tilde q_i} \tilde q_i + \sum_{i:q_i \leq \tilde q_i} \tilde q_i =1 \\
 \sum_{i:q_i > \tilde q_i} q_i - \tilde q_i &= \sum_{i:q_i \leq \tilde q_i} \tilde q_i - q_i,
\ee
while the fourth inequality occurs since for $i=\{1,\cdots,k-1\}, q_i\leq\tilde q_i$ is not included in the summation.

Now, to prove (iii) we need to construct a valid channel, which is characterized by the transition probabilities $P(i\rightarrow j)$ going from state $i$ to state $j$, where for all $i\in\{1,\cdots, k\}$,
\be
\label{validchannel}
\sum_{j=1}^k P(i\rightarrow j) = 1.
\ee 

To understand how we construct an appropriate channel, note that the channel has to increase slightly the probabilities $q_i$ to $\tilde q_i$ for all $i\in\{1,\cdots,k-1\}$, while decreasing $\tilde q_k$ accordingly for the normalization to hold. 

Bearing this in mind, let us denote the set $\mathcal{I}=\{i|1\leq i\leq k-1\}$ which contains the indices where $\tilde q_i \geq q_i$. Also define $\Delta_i=\tilde q_i - q_i$ and the sum $\Delta=\sum_{i\in\mathcal{I}} \Delta_i$. Note that $\tilde q_k = q_k - \Delta$ comes from the normalization of $\tilde{q}$. 

Now, let us consider the channel $E$ with the following transition probabilities:
\be
p(i\to j)= \left\{
\begin{array}{llll}
\label{channelconstruct}
1 &\text{for} &i=j,&  i,j \in \mathcal{I} \\
0 &\text{for}& i\not=j,&  i,j\in  \mathcal{I} \\
1- \frac{\Delta}{q_k}& \text{for} &i=j=k \\
0 &\text{for}& j=k, &  i\in  \mathcal{I} \\
\frac{\Delta_j}{q_k}, &\text{for}& i=k, & j\in\mathcal{I}.
\end{array}\right.
\ee

One can verify that the above transition probabilities satisfies \eqref{validchannel}, hence they characterize a valid channel. 
The last step is to show that this channel causes only a small perturbation to any other state $p$.
\be
||p - E(p)|| = p_k - \tilde p_k = \frac{\Delta}{q_k} p_k \leq \frac{\Delta}{q_k}.
\ee
Note that $\Delta = q_k - \tilde q_k \leq \frac{k}{N}$, hence
\be
||p-E(p)||\leq \frac{k}{N q_k} \leq \varepsilon \frac{\sqrt{N}}{\sqrt{k}} \frac{1}{\sqrt{k}} = \frac{\sqrt{\varepsilon}}{\sqrt{k}},
\ee
which holds since $\varepsilon = \frac{k}{N}$, and we have chosen $N$ so that $q_k > \frac{k}{\sqrt{N}}$.
\end{proof}

\begin{lemma}\label{lem:direct_sum}
Consider a channel such that for some fixed $n$-dimensional probability distribution $t=(t_1,\ldots,t_l, 0, \ldots 0)$ 
the channel gives output $t'=(t_1',\ldots,t_l', 0, \ldots 0)$. Moreover, 
$\Lambda(w)=w$ holds for some full rank distribution $w$. Then this channel can be decomposed as $\Lambda=\Lambda_1 \oplus \Lambda_2$,
where $\Lambda_1$ acts only on the first $l$-elements, mapping them onto the same group of elements, while $\Lambda_2$ acts similarly on the remaining $n-l$ elements.
\end{lemma}

\begin{proof}
Consider the joint probability of two random variables $(X,Y)$, given by the preserved distribution $w$ and the channel. Let $X=0$ denote the event that the inputs are from the first group (items from $1$ to $l$), and $X=1$ that 
they are from the second group (items from $l+1$ to $n$), and $Y$ denotes similar events for the outputs. 
Since the channel preserves $w$, we have that $P(Y=0)=P(X=0)$. Moreover, since the channel sends $t$ into $t'$ , this
means that $p(Y=0|X=0)=1$.
We then have
\be
p(Y=0)=p(Y=0|X=0)p(X=0)+ p(Y=0|X=1)p(X=1)= p(Y=0)+p(Y=0|X=1)p(X=1)
\ee
so that either $P(X=1)=0$ or $p(Y=0|X=1)=0$. However, since $w$ is of full rank, we know that $P(X=1)>0$.
Therefore $p(Y=0|X=1)=0$ must hold. We have therefore that $P(Y=0|X=0)=P(Y=1|X=1)=1$ 
which means that the channel is direct sum of two channels, acting on two disjoint groups of elements $\lbrace 1,\cdots,l\rbrace$ and $\lbrace l+1,\cdots,n \rbrace$.
\end{proof}

\subsection{Catalytic $d$-majorization}
\label{sec:ruch_trumping}

In this section, we prove a crucial result (Theorem \ref{th:catalysis}), which relates monotonicity of R{\'e}nyi divergences to catalytic transformations. This can be viewed as both a generalization of trumping relations \cite{Klimesh-trumping,Turgut-trumping-2007} and the $d$-majorization result \cite{RuchSS-1978}. It also gives an operational interpretation to the R{\'e}nyi divergences, answering the question posed in \cite{timerven2010}.

With the tools listed in Section \ref{ss:dmaj_catalysis_tools} in place, we can now proceed to state and prove the main theorem of this section.

\def\catalysis{
Given probability distributions $p,p',q,q'$, where $q, q'$ has full rank. The following conditions are equivalent:
\bei
\item[(i)] 
For all $\alpha \in (-\infty, \infty), D_\alpha(p || q)\geq D_\alpha(p' || q').$
\item[(ii)]
For any $\varepsilon >0$, there exists probability distributions $r,s$ of full rank, a distribution $p_{\varepsilon}'$ and a stochastic map $\Lambda$ such that 
\begin{enumerate}
\item $\Lambda(p\ot r) = p_{\varepsilon}'\ot r$,
\item $\Lambda(q\ot s) = q'\ot s$,
\item $\Vert p' - p_{\varepsilon}' \Vert \leq \varepsilon$.
\end{enumerate}
Moreover, we can take  $s=\eta$, where $\eta$ is  the uniform distribution onto the support of $r$.
\eei
}

\begin{theorem}
\catalysis
\label{th:catalysis}
\end{theorem}

\begin{proof} 
"(i) $ \to $  (ii)". To prove in this direction, we suppose that condition (i) holds, and construct a channel $\Lambda$ that satisfies (ii), with some $\varepsilon$ that can be chosen arbitrarily small. 
Firstly, choose some $\delta>0$, and define $p'' = (1-\delta) p' + \delta q'$. Then $\|p'-p''\|\leq\delta$. Using Lemma \ref{strictaux} together with (i), we find that for all $\alpha\in (-\infty, \infty)$,
\be \label{monotonicitywithptwoprimesstrict}
D_{\alpha}(p || q) > D_{\alpha}(p'' || q').
\ee

Now let us consider the following two cases separately:
\vspace{0.3cm}\\
\textit{(A) The probabilities in $q$ and $q'$ are rational.}
Without loss of generality, they can be written as $q = (d_1/N, \ldots d_k/N)$ and $q' = (d_1'/N, \ldots, d_k'/N)$ for some sets of integers $\{d_i\}_{i=1}^k, \{d'_i\}_{i=1}^k$, such that $\sum_{i=1}^k d_i =\sum_{i=1}^k d'_i=  N$.

With this, define two embedding channels $\Gamma$ and $\Gamma'$ associated with $\{d_i\}_{i=1}^k$ and $\{d_i'\}_{i=1}^k$ respectively. The fine-grained distributions of $p$ and $p''$ are given by 
\begin{align*} 
\tilde p &= \Gamma (p) = \oplus_i~p_i \eta_i\\
\tilde p'' &= \Gamma' (p'') = \oplus_i~p''_i \eta_i' ,
\end{align*}
where $\eta_i,\eta_i'$ are maximally mixed distributions of dimensions $d_i$ and $d'_i$ respectively. 
From \eqref{monotonicitywithptwoprimesstrict}, Lemma \ref{lem:rel_ent_H} tells us that
\be
D_{\alpha}(\tilde{p}||\eta_N)= D_\alpha(p||q) >  D_{\alpha}(p'' || q') =D_{\alpha}(\tilde{p}'' || \eta_N),
\ee
and hence by Proposition \ref{trumping}, we know that the distribution $\tilde p$ can be trumped into $\tilde p''$, i.e. there exists a probability distribution $r$ (the catalyst) and a bistochastic map $\Phi$ such that 
\be
\Phi(\tilde p \ot r) = \tilde p'' \ot r.
\ee
Note that by Lemma \ref{lem:direct_sum}, $r$ can be without loss of generality be of full rank, or in other words, the zeros in $r$ do not affect finding such a bistochastic map. More precisely, let us consider the case where $r$ does not have full rank, but has some rank $d$ instead. Then by setting $t=\tilde{p} \otimes r$, $t'=\tilde{p}''\otimes r$ and $w=\eta_N\otimes\eta$, we can use Lemma \ref{lem:direct_sum} to show that $\Phi=\Phi_1\oplus\Phi_2$, where the channel $\Phi_1$ gives us
\begin{equation}
\Phi_1(\tilde p \ot r) = \tilde p'' \ot r, ~~~~\Phi_1(\eta_N\ot \eta_d)=\eta_N\ot \eta_d,
\end{equation}
where $\eta_d$ is the uniform distribution on the support of $r$. On the other hand if $r$ is of full rank, then $\Phi=\Phi_1$.

Now, consider the following mapping 
\be
\tilde\Lambda =     (\Gamma'^* \ot \id )   \circ \Phi_1 \circ  (\Gamma \ot \id).
\ee
This map $\tilde\Lambda$ transforms $p\ot r $ into $p'' \ot r$, setting $\varepsilon=\delta$ satisfies conditions 1 and 3 in (ii). To satisfy condition 2, we want $\Lambda (q\ot s) = q'\ot s$ for some $s$. This is achieved by taking $s=\eta$, a uniform distribution of any dimension. Indeed, since $\Gamma (q) = \id/N$, $(\Gamma \ot \id) (q \ot \eta)$ is also a maximally mixed distribution. Since $\Phi_1$ is bistochastic, it preserves this distribution. Finally, by definition of $\Gamma'$ we have ${\Gamma'}^* (\id/N) = q'$.
\vspace{0.3cm}\\
\textit{(B) The distributions $q$ or $q'$ contain irrational values.}
We show that in such cases, a similar approach in (A) can be used, by considering distributions $\tilde q$ which are rational and close to the original distributions. 
Note that for any real $\alpha>0$, $D_\alpha (p\|q)$ is a continuous function of both arguments $p~{\rm and}~q$, whenever $q$ is of full rank. For $\alpha<0$, whenever $p$ does not have full rank, both $D_\alpha(p\|q)$ and $D_\alpha(p\|\tilde{q})$ diverge to infinity. When $p$ has full rank, continuity can be obtained by noting that $D_\alpha(p\|q) = c\cdot D_{1-\alpha} (q\|p)$ and $D_\alpha(p\|\tilde q) = c\cdot D_{1-\alpha} (\tilde q\|p)$ for some positive $c$, and $1-\alpha>0$, hence $D_{1-\alpha} (q\|p)$ is continuous.

From the above discussion, we conclude that when $\|q-\tilde q\|\leq \varepsilon$, in the limit of $\varepsilon\rightarrow 0$, \eqref{monotonicitywithptwoprimesstrict} implies that $D_\alpha( p || \tilde q) \geq D_\alpha( p'' || \tilde q')$ holds as well. 

By Lemma \ref{lem:channelrevised}, we use the channel $E$ that maps $q$ into $\tilde q$ while not perturbing $p$ too much. More precisely, for any $\varepsilon>0$, one can define a stochastic map $E$ such that 
\begin{equation}
\label{eq:E}
E(q) = \tilde q, \hspace{0.4 cm} \Vert \tilde q - q\Vert \leq \varepsilon,
\end{equation}
and $\Vert E(p) - p  \Vert \leq O(\sqrt{\varepsilon})$ for any other state $p$.
Thus Eq. (\ref{monotonicitywithptwoprimesstrict}) implies that in the limit of $\varepsilon\rightarrow 0$,
\be
D_{\alpha}( E(p) || \tilde q) =D_{\alpha}(p || q) > D_{\alpha}(p'' || q') = D_{\alpha}(p'' || \tilde q').
\ee

Following the first part of the proof (which established the result for rational $q$ and $q'$) we find that there is a catalyst $r$ and a stochastic operation $\Lambda$ such that
\begin{equation}
\Lambda( E(p) \otimes r) = p'' \otimes r, \hspace{0.4 cm} \Lambda( E(q) \otimes \eta) = \tilde q' \otimes \eta,
\end{equation}
with $\eta$ the maximally mixed distribution.

If $q'$ also contains irrational values, we can similarly define a second correction map that maps $q'$ into $\tilde q'$ while not perturbing $p''$ too much. By invoking Lemma \ref{lem:channelrevised} on $q'$, we construct $E'$ such that
\begin{equation}
E'(\tilde q') = q', \hspace{0.4 cm} \text{and for any probability distribution } r,~ \Vert E'(r) -r \Vert \leq O(\sqrt{\varepsilon}).
\end{equation}
Denoting the inverse of $E'$ as $E'^*$, hence $E'^*(\tilde q') = q'$. Choose an $r=E'^*(p'')$, then $\Vert E'(r) -r \Vert = \| p''-E'^*(p'') \| \leq O(\sqrt{\varepsilon})$.

Our final stochastic map is given by $E'^* \circ \Lambda \circ E$, where
\begin{align}
(E'^* \otimes \id) \circ \Lambda \circ (E \otimes \id)(q \otimes \eta) &= (E'^* \otimes \id) \circ \Lambda (\tilde q \otimes \eta) \nonumber\\
&= (E'^* \otimes \id) (\tilde q'\otimes\eta) \nonumber\\
&= q' \otimes \eta,
\end{align}
and 
\begin{align}
(E'^* \otimes \id) \circ \Lambda \circ (E \otimes \id)(p \otimes r) &= (E'^* \otimes \id) \circ \Lambda (E(p)\otimes r) \nonumber\\
& = (E'^* \otimes \id) (p'' \otimes r)\nonumber\\
& = p'''\otimes r,
\end{align}
with $p''' = E'^*(p'')$ such that
\begin{equation}
\Vert p''' - p'\Vert \leq \|p''-p'\|+\|p'''-p''\| \leq \delta + O\left( \sqrt{\varepsilon} \right). 
\end{equation}
Since $\varepsilon$ and $\delta$ can be chosen arbitrarily, the first part of the theorem follows.  

"(ii) $ \to $  (i)". Suppose that for all $\varepsilon >0$ there exist probability distributions $r,s, p_{\varepsilon}'$ and a stochastic map $\Lambda$ such that 
\be
\Vert p' - p_{\varepsilon}' \Vert_1 \leq \varepsilon,
\ee
and
\be
\Lambda(p\ot r) = p_{\varepsilon}'\ot r, \quad \Lambda(q\ot s) = q'\ot s,
\ee
and the support of $s$  includes the support of $r$. 
Then by monotonicity of the Renyi divergences, 
\begin{equation}
D_{\alpha}( p_{\varepsilon}'\ot r || q' \ot s) \leq D_{\alpha}(p\ot r || q \ot s),
\end{equation}
which equals 
\begin{equation}
D_{\alpha}( p_{\varepsilon}' || q') + D_{\alpha}(r || s) \leq D_{\alpha}(p || q ) + D_{\alpha}(r || s),
\end{equation}
by additivity. Since both $r$ and $s$ are full rank, $D_{\alpha}(r || s)$ is finite, and 
can be subtracted from both sides. Lastly, we consider the limit $p'_{\varepsilon}\to p'$. Recall that as long as the second argument $q'$ has full rank, 
for any $\alpha>0$, $D_\alpha(p'||q')$ is continuous w.r.t. both $p'$ and $q'$. For $\alpha<0$, whenever $p'$ with full rank, continuity holds. If $p'$ does not have full rank, $ \lim_{\varepsilon\rightarrow 0} D_\alpha (p_{\varepsilon}'\|q') = \infty = D_\alpha (p'\|q')$. Hence we obtain
\begin{equation}
D_{\alpha}( p' || q')  \leq D_{\alpha}(p || q )
\end{equation}
for all $\alpha>0$ and $\alpha<0$. Since $D_0 (p\|q)=\lim_{\alpha\to 0^+} D_\alpha (p\|q)$, the above inequality holds also for $\alpha=0$.

\end{proof}

\subsection{For diagonal input state of the system, a diagonal catalyst is enough}

\label{subsec:diagonal}
Here we will show that if the initial state of the system is block diagonal in the energy eigenbasis, then the diagonal of the output state does not depend on coherences of the catalyst
but only on its diagonal. This means that if we are interested only in the diagonal 
of the output state of the system, we can replace the catalyst with its dephased version (i.e. put a diagonal 
catalyst that has the same diagonals as the original one). The conditional probabilities form the channel,
which maps the initial diagonal to the final diagonal of the state of the system. By block diagonal, we mean that the state can be written as
\be
\rho=\sum \sigma_{ijk}\ket{E_i,j}\bra{E_i,k}
\ee
where $\ket{E_i,j}$ has energy $E_i$ and has degenerate levels labled by $j$.
 
To see that only block diagonal catalysts are needed, we write the initial state as
\be
\rho^{in}_{RSC}=\rhoinr\ot\rhoins\ot\rhoinc,
\ee
where $\rhoinr$ is the heat bath, which is of course diagonal, $\rhoinc$ is the state of an arbitrary catalyst, 
and $\rhoins$ is the state of the system which we assume to be diagonal.
We then act with an energy-preserving unitary $U$ and get the output state
\be
\rho^{out}_{RSC} = U \rho^{in}_{RSC}  U^\dagger.
\ee

We now compute the diagonals element of $\rho^{out}_S$, and will see
that they depend only on the diagonal elements of $\rhoinc$. We have
\be
\<E_S|\rhoouts|E_S\>=\sum_{E_R,E_C} \<E_R,E_S,E_C|\rho^{out}_{RSC}|E_R,E_S,E_C\>
\ee
(here and in the following we sum over energies as well as degeneracies).
This can be written as 
\ben
&&\<E_S|\rhoouts|E_S\>=\sum_{E_R,E_C} \sum_{E_R',E_C',E_S',E_C''}
\<E_R,E_S,E_C|U|E_R'E_S'E_C'\>\<E_R',E_S',E_C'|\rho^{in}_{RSC}|E_R',E_S',E_C''\>\times\nonumber \\
&&\times \<E_R',E_S',E_C''|U^\dagger|E_R,E_S,E_C'\>
\een
where we used that $\rhoinr$ and $\rhoins$ are diagonal. Since $U$ preserves energy, 
we have $E_R+E_S+E_C=E_R'+E_S'+E_C'$ as well as $E_R+E_S+E_C=E_R'+E_S'+E_C''$. This implies that $E_C'=E_C''$. 
We thus obtain
\be
\<E_S|\rhoouts|E_S\>= \sum_{E_S'} p(E_S|E_S') \<E_S'|\rhoouts|E_S'\>,
\ee
where the conditional probabilities are given by 
\be
p(E_S|E_S')= \sum_{E_R,E_C}\sum_{E_R',E_C'}|\<E_R,E_S,E_C|U|E_R',E_S',E_C' \>|^2 \<E_R'|\rhoinr|E_R'\>  
\<E_C'|\rhoinc|E_C'\>.
\ee
One easily finds that indeed this is a valid conditional probability distribution. 

\section{The Second Laws}
\label{sec:secondlaws}

In this section we formulate the state transformation conditions, namely the second laws of thermodynamics. We will first do this for states diagonal in the energy eigenbasis,
and then for general state transformations. The former states are those $\rho$ which satisfy $[\rho,H]=0$. 
The conditions are given in terms of \genFs, which are defined as follows:
\be
F_\alpha(\rho,H)=-kT [\ln Z + D_\alpha(\rho\|\rho_\beta)]=F(\rho_\beta, H) + kT D_\alpha(\rho\|\rho_\beta)
\ee
where $Z$ is the partition function for the Hamiltonian $H$, and $\rho_\beta$ is the thermal state. Note that $F_1$ 
is the standard free energy, and that for thermal states and pure eigenenergy states, all free energies are 
equal to the standard one.

In Subsection \ref{ss:qsecondlaw} we will present the quantum limitations which apply to all state transformations.

\subsection{The second laws for states block diagonal in energy}

\begin{theorem}[Second laws for block diagonal states]
Consider a system with Hamiltonian $H$. Then
a state $\rho$ block diagonal in the energy eigenbasis can be transformed with arbitrary accuracy into another block diagonal state $\final$ under catalytic thermal operations if and only if, for all
$\alpha \in (-\infty, \infty)$,
\begin{equation}
F_\alpha( \rho,H) \geq F_\alpha( \final,H).
\label{eq:second_laws}
\end{equation} 
\label{thm:second_laws}
\end{theorem}

\begin{remark}
By continuity, the conditions can equivalently include $\alpha=\infty,-\infty$. 
\end{remark}

\begin{proof}
We will prove it using Theorem \ref{th:catalysis}. Consider initial and final states of the system $\rho_S$ and $\rho_S'$.  Suppose first that for $\alpha \in (-\infty,\infty)$, equation \eqref{eq:second_laws} holds, which is equivalent to 
\be
D_\alpha(\rho_S || \sgibbs)\geq D_\alpha(\rho_S' || \sgibbs).
\ee 
We now need to show that one can transform the state $\rho_S$  into an arbitrary good version of $\rho_S'$ by catalytic thermal 
operations. 

We denote  $p$ and  $p'$ to be eigenvalues of $\rho_S$ and $\rho'_{S}$ respectively, and $q=q'=\sgibbs\ot \ancgibbs$.
Then, using Theorem \ref{th:catalysis}, we get that there exists a channel $\Lambda$ and uniform distribution $\eta$ such that
(i) $\Lambda$ preserves the state $q\ot \eta$, (ii) $\Lambda$ sends $p\ot r$ into $p'_\epsilon\ot r$,
where $p'_\epsilon$ approximates $p'$. Condition (i) means that we can take 
the catalyst system with trivial Hamiltonian. Thus we have 
\be
\Lambda(\rho_S \ot \rho_C)=\rho^{\rm out}_{S}\ot \rho_C,
\ee
where $||\rho^{\rm out}_{S}-\rho_S'||\leq \epsilon$ and
\be
\Lambda(\sgibbs\ot\ancgibbs \ot\cgibbs)=\sgibbs\ot\ancgibbs\ot \cgibbs,
\ee
where $\cgibbs$ is the maximally mixed state on catalyst system $C$, 
i.e. $\Lambda$ preserves the thermal state of the system $SC$ and the ancilla. 
However thermal operations are precisely the operations that preserve the thermal state 
\cite{HO-limitations,Beth-thermo}. Thus the required transition can be made by catalytic thermal operations.

Conversely, let us assume that for given states $\rho_S$ and $\rho_S'$ there exists a quantum channel $\Lambda$,
and a system $C$ with the hamiltonian $H_C$, and state $\rho_C$ such that 
\be
\Lambda(\sgibbs \ot \cgibbs) = \sgibbs \ot \cgibbs, \quad \Lambda(\rho_S\ot\rho_C)=
\rho_{S}^{out}\ot\rho_C.
\ee
where $||\rho_{S}^{\rm out}-\rho_S'||\leq \epsilon$. 
Now, let us note that by thermal operations, one can make any state diagonal in the energy basis,
namely one can apply random phases (by using the Birkhoff primitive of \cite{thermoiid} 
which supplies random noise and unitaries that change phases). 
In Sec. \ref{subsec:diagonal} we have shown that since the input and output states are diagonal,
we can take the state $\rho_C$ to be diagonal too, and the above condition reads
\be
\Lambda_{\rm cl}(q\ot s)= q\ot s, \quad \Lambda_{\rm cl}(p\ot r)= p'_\epsilon\ot r,
\label{eq:Lambda_rs}
\ee
Thus we can apply theorem \ref{th:catalysis}, obtaining that 
\be
D_\alpha(\rho_S\|\sgibbs)\geq D_\alpha(\rho'_S\|\sgibbs)
\ee
for all real $\alpha$ which is equivalent to \eqref{eq:second_laws}.
\end{proof}

\subsection{Getting rid of second laws with negative $\alpha$}
In the paradigm of catalytic thermal operations it is assumed that the
catalyst is returned exactly. If we had required that the catalyst is returned only with good fidelity, 
all restrictions would be lifted, and there will be no conditions for thermodynamic transitions at all. 
This happens when the catalyst's dimension is free to be arbitrarily large. 
In such cases, fidelity is not a good criterion for closeness any more (we discuss this in more detail in Section \ref{sec:inexact} of the \supl). 

However, in addition to the catalyst,  we can consider borrowing some system with fixed size (e.g. a qubit) in a pure state, given that we will return it with arbitrary good fidelity. As we will argue now, 
this will lift all the conditions on negative $\alpha$. However since creating a pure state from a thermal state without additional resource requires infinite work, we will also consider approximate versions. 

Here is the result, where we allow the use of an exactly, pure qubit. 
\begin{theorem}
\label{thm:second_laws2}
If we can use catalytic thermal operations, and 
are allowed to borrow a qubit 
with a trivial Hamiltonian in the state $|0\>\<0|$ and have to return it 
with arbitrary accuracy,  then  a state $\rho$ block diagonal in the energy eigenbasis can be 
transformed with arbitrary accuracy  into another block diagonal state  
$\final$ if and only if, for all 
$\alpha \geq 0$
\begin{equation}
F_\alpha( \rho,H) \geq F_\alpha( \final,H).
\label{eq:second_laws2}
\end{equation} 
\end{theorem}

\begin{proof}
Suppose first that
one can transform the state $\rho_S \ot |0\>\<0|$ into arbitrary good version of $\rho_S' \ot |0\>\<0|$.
Then we proceed as in the proof of Theorem \ref{th:catalysis} using monotonicity, additivity of $D_\alpha$ and noticing that $D_\alpha$ is finite for 
$|0\>\<0|$ for $\alpha>0$, we get
\be
D_\alpha(\rho_S || \sgibbs)\geq D_\alpha(\rho_S' || \sgibbs)
\label{eq:second_laws_D}
\ee 
for $\alpha \geq 0$, so that conditions \eqref{eq:second_laws2} are satisfied.  
Conversely, assume that the latter conditions are satisfied (hence equivalently \eqref{eq:second_laws_D} hold). 
However for  
 $\alpha<0$, $D_\alpha$ is infinite on the left side, and finite on the right, if for the output side we return the qubit in a state which is of full rank, but arbitrarily close to it's original state. 
 Moreover $D_\alpha(\rho_S\ot |0\>\<0|~|| {\gibbs}_{SC})=D_\alpha(\rho_S\| {\gibbs}_S)+ D_\alpha(|0\>\<0|~|| {\gibbs}_C)$ 
and same for $S'$.  Thus we get that for all real $\alpha$ 
\be
D_\alpha(\rho_S || \sgibbs)\geq D_\alpha(\rho_S' || \sgibbs).
\label{eq:second_laws_D2}
\ee
and by Theorem \ref{thm:second_laws} we can transform the state $\rho_S \ot |0\>\<0|$ into arbitrary good 
version of $\rho_S' \ot |0\>\<0|$. 
\end{proof}


We thus borrow a pure state qubit, and return it with full rank, but arbitrarily close to being pure, to get rid of the conditions for negative $\alpha$.
However, let us note, that we might borrow a noisy version of $|0\>\<0|$
namely $(1-\epsilon)|0\>\<0|+ \epsilon |1\>\<1|$. 
It costs a finite amount of work to create it from the thermal state. If $F_\alpha(\rhoins, H)\geq F_\alpha(\rhoouts, H)$ 
for $\alpha\geq 0$, we can pretend that the ancilla is in pure state $|0\>\<0|$ 
and apply the map from Theorem \ref{thm:second_laws2}. We will obtain the output state arbitrarily close to
\be
(1-\epsilon) \rhoouts\ot |0\>\<0|_{\rm anc} + \epsilon \rho'_{\rm S,anc}.
\ee
We can then return the ancilla in the same state (by depolarizing a bit if needed), 
and obtain the output state $\rhoouts$  with accuracy $\epsilon$ in trace norm. 
However as we will discuss in section \ref{sec:inexact}, if there are no further restrictions on the available catalyst, then closeness in trace distance is not a 
suitable demand on the returned catalyst in thermodynamic transformations.

%

\subsection{Borrowing ancilla with a nontrivial Hamiltonian}

There is another approach, that might eliminate some of the conditions \ref{eq:second_laws}, 
yet is again not acceptable. Namely, we can borrow a qubit with nontrivial Hamiltonian $H=E|E\>\<E|$ 
in state $\rho_\epsilon=(1-\epsilon)|0\>\<0| + \epsilon |E\>\<E|$
and return thermal state $\rho_\beta$. If we take $\ep\geq \frac{e^{-\beta E}}{Z}$, we have $||\rho_\epsilon-\rho_\beta||_{1}\leq \ep$.
On the other hand, 
for any $\alpha>1$, if we set e.g. $\epsilon=1/E^2$ and use the approximation that $E$ is large, then we have
\be
D_\alpha(\rho_\epsilon\|\gibbs) &= \frac{1}{\alpha-1} \log \left[\left(1-\frac{1}{E^2}\right)^\alpha \left(1-\frac{e^{-\beta E}}{Z}\right)^{1-\alpha} + \left(\frac{1}{E^2}\right)^\alpha \left(\frac{e^{-\beta E}}{Z}\right)^{1-\alpha}\right]\\
& \gtrsim \frac{1}{\alpha-1}\log \left[\left(\frac{1}{E^2}\right)^\alpha \left(\frac{e^{-\beta E}}{Z}\right)^{1-\alpha}\right]\\
& \geq \frac{\alpha}{\alpha-1} \log \frac{1}{E^2} - \log \frac{e^{-\beta E}}{Z}\\
&\geq \frac{1}{\ln 2} \beta E - \frac{2 \alpha}{\alpha-1} \log E,
\ee
which diverges for large $E$. Here, we used the approximation that the first term in the logarithm does not contribute much, since $\left(\frac{e^{-\beta E}}{Z}\right)^{1-\alpha}$ is large for $\alpha>1$. Furthermore, $Z=1+\frac{e^{-\beta E}}{Z}\approx 1$. So, if we consider the transition 
\be
\rhoins\ot \rho_\epsilon \to \rhoouts \ot \rho_\beta,
\ee
then for any fixed $\alpha_0>1$, if we pick high energy $E$, we obtain that the $D_\alpha$ 
of the left hand side can be arbitrarily large, and the transition is possible, once conditions 
for $\alpha<\alpha_0$ are satisfied.  
However in this approach, we are allowed to borrow something very expensive, in the sense that it requires a
large amount of work to be created from the thermal state, and return something useless - a thermal state. 
In other words there is a large {\it work distance} between the state we return and the state we have borrowed 
(see section \ref{sec:workdistance}).

It is interesting to compare the two situations:
to eliminate negative $\alpha$ we borrow a pure ancilla state. 
The Hamiltonian of the ancilla is trivial, and the state is much more pure than the thermal one.
That's why it is (infinitely) expensive. 
To eliminate $\alpha>1$, we have the opposite: we take a Hamiltonian of the ancilla, such that thermal state is very pure, 
and we borrow a state which is {\it more mixed} than the thermal state, and for this reason also expensive. 
If we borrow both ancillas, we are left with conditions $\alpha\in [0,1+\delta]$, with $\delta$ depending 
on how large an energy we will have in one of the ancillas.

\subsection{Quantum second laws: limitations for states that are not diagonal in energy basis}
\label{ss:qsecondlaw}

For states that are not diagonal in energy basis, we report the following limitations:
\begin{proposition} [Quantum limitations]
For a state $\rho$ to be transformed with arbitrary accuracy into another state $\final$ 
(both not necessary diagonal in energy basis) we require
\bei
\item[i)]For $\alpha\geq \frac12$,
\be
\qrenyi_\alpha(\rho || \gibbs) \geq \qrenyi_\alpha(\final||\gibbs),
\ee
\item[ii)] For  $\frac12 \leq \alpha\leq 1$,
\be
\qrenyi_\alpha(\gibbs || \rho) \geq \qrenyi_\alpha(\gibbs||\rho'),
\label{eq:cond2_second_laws}
\ee
\item[iii)] For  $-1 \leq \alpha\leq 2$,
\be
\qrenyisimple_\alpha(\rho||\gibbs ) \geq \qrenyisimple_\alpha(\rho||\gibbs).
\label{eq:cond3_second_laws}
\ee
\eei

\end{proposition}

This proposition follows directly from monotonicity under 
completely positive and trace-preserving maps, and additivity of the quantum R{\'e}nyi divergences associated
with the generalised free energies, as discussed in Section \ref{ss:quantumdivergence}.  
E.g. we have 
\be
\qrenyi_\alpha(\rho~|| \gibbs) &\geq \qrenyi_\alpha(\Lambda(\rho) || \Lambda(\gibbs))\\
& \geq  \qrenyi_\alpha(\rho' || \gibbs)\\
\ee
with the first inequality due to monotonicity of $\qrenyi_\alpha$ under completely positive trace-preserving maps $\Lambda$, and the second line following
from the fact that if $\Lambda(\rho)=\rho'$ is a thermal operation, it preserves the thermal state.
The reason for considering condition \eqref{eq:cond2_second_laws} only for $\frac12 \leq \alpha \leq 1$,
is that for $\alpha\geq 1$, the entropies would diverge for non full rank catalyst, so 
may not be monotones if we use such  catalysts. As before, by using a pure state which is returned arbitrarily close to pure,
we will automatically satisfy Equation \eqref{eq:cond3_second_laws} for $\alpha\leq 0$ and we may thus only consider $0\leq\alpha \leq 2$.

These conditions are like second laws, in that they are necessary conditions 
which must be respected during a
thermodynamic transition, however, they are not sufficient. This follows from the fact that there are operations which preserve
the thermal state, which are not thermal operations, and which can transform a state which is block diagonal in the energy
eigenbasis, to one which is not~\cite{FOR-inprep}. Such a transition is impossible by thermal operations since thermal operations cannot
transform a block diagonal state to one which is not. However, such a transition will respect any second law which also hold 
for operations which preserve the thermal state, as is the case with our quantum second laws.


\subsection{Application: Landauer erasure with a quantum Maxwell demon}
\label{ss:erasure}

As an application of the above second laws, we consider a special case of a Maxwell demon with a memory $Q$
in state $\rho_Q$, who wants to reset a system $S$ in state $\rho_S$ to some pure state. One imagines that the demon's 
memory is correlated in some way with the system so that the total state of demon plus system is $\rho_{SQ}$. The demon wishes
to reset the  state of the system by transforming it to some pure state, but the state of the demon's memory should not change.
This can be seen as a fully quantum version of the standard Maxwell demon/Landauer erasure scenario, and has been considered in the case of a 
trivial Hamiltonian\cite{del2011thermodynamic,faist2012quantitative} and when the Maxwell demon does not have access to
 ancilliary systems. The result of \cite{del2011thermodynamic} gave a thermodynamic
interpretation to the notion of negative information~\cite{how-merge}. We will see here, that the amount of work
which is required for such an operation, is quantitatively different when all thermodynamical operations are considered. This suggests that
in the single-shot scenario, the notion of partial information takes on a different form, which we shall now derive. 

In fact, the result follows immediately from the work distance of Equation \eqref{eq:trumpd}, namely, we have that the cost of resetting the system $W_{\rm reset}$ to the 
pure state $\psi_{E}$ with energy $E$ is given by
\begin{align}
W_{\rm reset}&= \inf_{\alpha} [D_\alpha(\rho_{QS}\|{\gibbs}_{QS})-D_\alpha(\rho_Q\otimes\psi_{E}\|{\gibbs}_{SQ})]\nonumber\\
&=\inf_{\alpha} [D_\alpha(\rho_{QS}\|{\gibbs}_{QS})-D_\alpha(\rho_Q\|{\gibbs}_{SQ})]+E
\label{eq:divergence-difference}
\end{align}
In the case when the Hamiltonian is trivial, this reduces to
\begin{align}
W_{\rm reset}=\inf_{\alpha} [H_\alpha(\rho_{QS})-H_\alpha(\rho_Q)]
\label{eq:renyi-difference}
\end{align}
When this quantity is negative, the demon can reset the system to a pure state, and not only does this not cost work, but the corresponding amount of work
is actually gained. Not only can thermodynamical work be negative, as shown in \cite{del2011thermodynamic}, but it can be even more negative! In the sense that catalytic operations can be used to gain more work than would be otherwise possible. Note that this gives the difference of Renyi entropies of
Equation \eqref{eq:renyi-difference} an operational interpretation, very similar to that enjoyed by the conditional von-Neumann entropy in the case
of identically and independently distributed states.


\section{Approximate catalysis}
\label{sec:inexact}
So far we considered exact catalysis, i.e. the process is perfectly cyclic and the catalyst should be in the same state as it initially was. 
However, this is usually an unreasonable demand because in physical processes there are some unavoidable inaccuracies. 
Therefore, we ask what are the conditions for transformations, when the 
catalyst $\rhoinc$ can be returned in a state $\rhooutc$ that is merely 
''close'' to the initial state.

\subsection{How to quantify ''approximate''}
As already discussed in the introduction, it turns out that what is meant by 
"close" matters greatly. 
%
At first glance, one might be tempted to demand that close should mean that $\rhoinc$ is close to $\rhooutc$ in terms of the fidelity, or essentially equivalently, the trace norm distance. The trace norm distance between two states $\rho$ and $\sigma$ can be written as
\begin{align}
\|\rho - \sigma\|_1 = \max_{0 \leq M \leq \id} \tr\left[M\left(\rho-\sigma\right)\right]\ .
\end{align}
We say that $\rho$ and $\sigma$ are $\epsilon$-close if $\|\rho-\sigma\|_1 \leq \epsilon$.
It enjoys an appealing operational interpretation as being  
$\epsilon$-close in trace norm distance means that if we were given 
states $\rho$ and $\sigma$ with probability $1/2$ each, then our probability of correctly distinguishing them by any physically allowed measurement
is bounded by $1/2 + \epsilon/2$. In other words, being close in the trace distance means that the two states cannot be distinguished well by any physical process
\footnote{
Finally, we note that when considering how cyclic a process is, we might also ask about preserving correlations between the catalyst $\rhoinc$ and its purifying systems. Clearly, there exist simple operations that preserve the catalyst even exactly, and yet do not preserve correlations with separated reference systems holding the purification. Consider for example a catalyst of the form $\rhoinc = \rho_1 \otimes \rho_2$ where $\rho_1 = \rho_2$ and the purifications of $\rho_1$ and $\rho_2$ are held separately by systems $1$ and $2$ respectively. Returning the catalyst as $\rho_2 \otimes \rho_1$ preserves the catalyst, but it requires interaction between the two purifying systems of $\rho_1$ and $\rho_2$ in order to restore the correlations with the catalyst. Such a paradigm is more restrictive than the one considered here.
}. 
In terms of the catalyst, one might hence ask 
that $\|\rhoinc - \rhooutc\|_1 \leq \epsilon$ for some arbitrary small $\epsilon$. This would be a mistake as we will see.

\subsection{The embezzling state dilemma: when there is no second law}
\label{subsec-embezzling}

As mentioned already in the introduction, there is the phenomenon of embezzling \cite{Hayden-embezzling} in entanglement theory, where at an expense of increasing the size of the catalyst, one can perform arbitrary transformation with arbitrary good fidelity, while returning the catalyst in a state arbitrarily close to the initial state. 
More specifically, for $H=0$ 
we can adapt the results of~\cite{Hayden-embezzling} to show that for any $\epsilon$, there exists a dimension $d$ such that using the catalyst
\begin{align}
\rhoinc = \frac{1}{C(n)}\sum_{j=1}^n \frac{1}{j} \proj{j}\ , 
\end{align}
allows us to transform any state $\initial$ $\epsilon$-close 
to any state $\final$ such that
$\|\rhooutc  - \rhoinc\|_1 \leq \epsilon$. Here, $C(d)$ is a normalisation constant. In particular, there exists an $n$ large enough, such that for any $\epsilon$,
we can use $\rhoinc$ to erase a state $\varphi=\sum_{i=1}^m p_i \proj{i}$ to $\ket{0}$ via a unitary transformation taking $\rhoinc\otimes\varphi$ to a state close to $\rhoinc\otimes\proj{0}$.
We first apply a unitary  on the state $\rhoinc\otimes\varphi$ to produce a state $\omega=\sum_{i=1}^{nm} q_i\proj{i}$ such that the eigenbasis of $\omega$ is the same as 
$\rhoinc\otimes\varphi$, and the probabilities $q_i$ are the same as the eigenvalues of $\rhoinc\otimes\varphi$, but in decreasing order.  One can then show~\cite{Hayden-embezzling} that $\sum_j\sqrt{q_j/jC(n)}$ goes to $1$ as $n$ goes to infinity, which implies that the fidelity $F(\omega,\rhoinc)=\tr\sqrt{\omega}\sqrt{\rhoinc}\rightarrow 1$. 
Since the fidelity 
upperbounds the trace distance via $F(\omega,\rhoinc)\geq1-\frac{1}{2}\|\omega-\rhoinc\|$ we have the desired result.

Another example of embezzling work occurs with $H=J$ and the pure state coherence resource used in \cite{thermoiid}
\begin{align}
\ket{\psi} = \frac{1}{\sqrt{d}}\sum_{j=1}^d  \ket{j} 
\end{align}
with $\ket{j}$ energy eigenstates.
If we apply the energy conserving operation which takes $\ket{0}_S\ket{j}_C\rightarrow|\ket{1}_S\ket{j-1}_C$ to the state $\ket{0}_S\ket{\psi}_C$,
then by increasing $d$, the catalyst will be left in a state arbitrarily close to it's original state $\ket{\psi}$ in trace distance, even though
one unit of work (energy) has been transferred to the system.

Demanding that 
the catalyst be returned close in the trace distance
is thus too weak a condition. Intuitively, the reason why it is too weak is
that we may still need to consume much work to obtain the original catalyst from its returned version.

If we therefore concieve of an approximately cylic process, as one which the working body is returned in a state which is $\epsilon$ close in fidelity to it's original
form, then there is no second law. All state transformations are possible.

\subsection{The work distance} 
\label{sec:workdistance}


We thus take an operational perspective on the problem of inexact catalysis. More precisely, we propose to take as a reference point, exact catalysis, and require that 
the state that is returned should require a small amount of work in order to be restored 
to its original state. This is natural, in the sense that if someone loans you a catalyst, then they would want it returned in such a way that it would not require a large amount
of work to return it back into it's original state. 
 
We thus consider the inexact transition 
\be
\rho_S^{\rm in} \ot \rho_C^{\rm in} \to \rho_{SC}^{\rm out} 
\ee
and require that to obtain from $\rho_{SC}^{\rm out}$ the original state of the catalyst and the required output state 
of the system, we need not input more than a small amount of work. 

This prompts our definition of the work distance below. To make this precise, let us take a closer look how one can derive upper bounds on the amount of work needed for state transformation, namely to restore the catalyst to its original form. To input this amount of work, we can append a battery to either provide or extract work, to facilitate this transformation. We then apply our conditions for state transformation to the state and the battery together. The battery we use is a two level system called a \textit{wit}, or {\it work bit}, introduced previously in \cite{HO-limitations}. The wit initially starts out in the energy eigenstate $\omega_i=|0\>\<0|$ which has energy $E_i=0$. At the end of the process, the battery is in another energy eigenstate $\omega_f=|W\>\<W|$, having energy $W$. $W$ can be either negative or positive, depending on whether work is used from or stored in the battery system. This means that the following transition is possible 
\be
\rho_{SC}^{\rm out} \ot |0\>\<0|_W  \to \rho_S^{\rm out} \ot \rho_C^{\rm in} \ot |W\>\<W|_W,
\ee
while the thermal state of the battery system is given by 
\begin{equation}\label{eq:tauw}
\tau_W = \frac{e^{-\beta W}|W\rangle\langle W|+|0\rangle\langle 0|}{1+e^{-\beta W}},
\end{equation}
with $\beta=\frac{1}{kT}$, $k$ being the Boltzmann constant and $T$ being the temperature of the bath that is in contact with the system and battery.
For this transformation to be possible, it is required in Theorem \ref{th:catalysis} that the following conditions hold for all $\alpha\geq0$,
\begin{align}\label{inexactcondition}
F_\alpha(\rho_S^{\rm in}\otimes\rho_C^{\rm in}\ot\omega_i\|{\gibbs}_S\otimes{\gibbs}_C\otimes\tau_W) &\geq F_\alpha(\rho_S^{\rm out}\otimes\rho_C^{\rm in}\ot\omega_f\|{\gibbs}_S\otimes{\gibbs}_C\otimes\tau_W).
\end{align}
Since the initial, final, and thermal state of the battery is known and depends only on the parameter $W$, we can then derive an upper bound for $W$ from \eqref{inexactcondition} and~\eqref{eq:tauw} as
\begin{align}\label{workdistderivation}
F_\alpha(\rho_S^{\rm in}\|{\gibbs}_S) + \frac{kT}{\alpha-1}\log\left[\frac{1}{1+\e^{-\beta W}}\right]^{1-\alpha} &\geq F_\alpha(\rho_S^{\rm out}\|{\gibbs}_S) + \frac{kT}{\alpha-1}\log\left[\frac{e^{-\beta W}}{1+\e^{-\beta W}}\right]^{1-\alpha}
\end{align}
Rearranging yields
\begin{align}
F_\alpha(\rho_S^{\rm in}\|{\gibbs}_S) &\geq F_\alpha(\rho_S^{\rm out}\|{\gibbs}_S)+ W\ ,
\end{align}
which yields the following upper bound on $W$
\begin{align}
W &\leq F_\alpha(\rho_S^{\rm in}\|{\gibbs}_S)-F_\alpha(\rho_S^{\rm out}\|{\gibbs}_S).
\end{align}
In essence, the possible transitions are governed by Renyi divergences, up to tolerance $W$, which we choose. Since $W$ has to be smaller than the above bounds for all positive alpha, the maximal amount of work extractable for such a process, going from $\initial$ to $\final$ will be given as
\begin{equation}
\trumpd(\initial\succ\final)= \inf_{\alpha>0}~[F_\alpha(\initial\|\gibbs)-F_\alpha(\final\|\gibbs)].
\end{equation}
where we refer to $\trumpd(\initial\succ\final)$ as the \textit{work distance} from $\initial$ to $\final$.

It is interesting to see how the conditions presented in~\cite{HO-limitations} arise as special cases of our conditions, and are hence independent of a catalyst.
In \cite{HO-limitations}, the extractable work from a state (by thermalizing it) and the work cost for its formation (starting with a thermal state) via thermal operations have been given by 
\begin{align}
\hat{W}_{\mathrm{ext}}(\rho) &= -kT \log \tr (\Pi_\rho \rho_\beta) = kT D_0(\rho\|\gibbs)\\
\hat{W}_{\mathrm{cost}} (\rho) &= kT \log \min\lbrace\lambda:\rho\leq\lambda\rho_\beta\rbrace= kT D_\infty(\rho\|\gibbs),
\end{align}
The subsequent corollary shows that the work distance reduces to these quantities.

\begin{corollary}
Consider initial and final states $\rho$ and $\rho'$:
\bei
\item If $\rho'$ is the thermal state, then the maximum extractable work $\hat{W}_{\mathrm{ext}}(\rho)=\trumpd(\rho\succ\rho')$.
\item If $\rho$ is the thermal state, then the minimum work cost $\hat{W}_{\mathrm{cost}}(\rho')=-\trumpd(\rho\succ\rho')$.
\eei
\end{corollary}
\begin{proof}
If $\rho'=\gibbs$, then $\forall\alpha,~D_\alpha(\rho'\|\gibbs)=0$, hence 
\begin{align*}
\trumpd(\rho\succ\rho') &=kT\cdot \inf_{\alpha>0}~D_\alpha(\rho\|\gibbs) = \hat{W}_{\mathrm{ext}}(\rho),
\end{align*}
where the last equality holds due to the fact that the R{\'e}nyi divergences are non-decreasing in positive $\alpha$.
If $\rho=\gibbs$, then similarly
\begin{align*}
-{\rm D} (\rho\succ\rho') &=-kT\cdot \inf_{\alpha>0}~[-D_\alpha(\rho'\|\rho_\beta)] =kT\cdot \sup_{\alpha>0}~[D_\alpha(\rho'\|\gibbs)]= \hat{W}_{\mathrm{cost}}(\rho').
\end{align*}
\end{proof}

Let us now return to the discussion of catalysis, where we demand that the catalyst is returned, such that the work distance for resetting the catalyst to it's original state is small. In the case where inexact catalysis occurs, we are allowed to borrow a catalyst $\rho_C^{\rm in}$ and conduct the transformation $\rho_S^{\rm in}\otimes \rho_C^{\rm in} \rightarrow \rho_{SC}^{\rm out}$. If this transformation is allowed via thermal operations, then we know for all $\alpha\geq 0$, 
\begin{equation}
F_\alpha (\rho_S^{\rm in}\otimes \rho_C^{\rm in},\rho_{\beta S}\otimes \rho_{\beta C}) \geq F_\alpha (\rho_{SC}^{\rm out},\rho_{\beta S}\otimes \rho_{\beta C}).
\end{equation}

Consider the restrictions as follows: if $\rho_{SC}^{\rm out}= \rho_{S}^{\rm out}\otimes \rho_{C}^{\rm out}$ is of product form(see our previous discussion at the start of Section \ref{sec:TO} on why this is required)\footnote{For non-product output distributions, this argument does not hold since the $\alpha$ free energies are not superadditive, namely $F_\alpha (\rho_{AB},\sigma_{AB})\geq F_\alpha(\rho_A,\sigma_A)+F_\alpha(\rho_B,\sigma_B)$ is generally not true.} , and if the cost of restoring the catalyst has to be small, $-{\rm D} (\rho_{C}^{\rm out}\succ\rho_{C}^{\rm in})\leq\varepsilon$, namely ${\rm D} (\rho_{C}^{\rm out}\succ\rho_{C}^{\rm in})\geq-\varepsilon$, then
\begin{align}
F_\alpha (\rho_S^{\rm in}\otimes \rho_C^{\rm in},\rho_{\beta S}\otimes \rho_{\beta C}) &\geq F_\alpha (\rho_{S}^{\rm out},\rho_{\beta S})+F_\alpha(\rho_{S}^{\rm out},\rho_{\beta C})\\
F_\alpha (\rho_S^{\rm in},\rho_{\beta S}) &\geq F_\alpha (\rho_{S}^{\rm out},\rho_{\beta S})+F_\alpha(\rho_{S}^{\rm out},\rho_{\beta C})- F_\alpha (\rho_C^{\rm in},\rho_{\beta C})\\	
&\geq F_\alpha (\rho_{S}^{\rm out},\rho_{\beta S})+ \inf_{\alpha>0} \left[ F_\alpha(\rho_{S}^{\rm out},\rho_{\beta C})- F_\alpha (\rho_C^{\rm in},\rho_{\beta C})\right]\\
&\geq F_\alpha (\rho_{S}^{\rm out},\rho_{\beta S})-\varepsilon,
\end{align}
which tells us that our second laws are recovered.

\subsection{Small error per particle -- recovering the free energy}
\label{sec:extensive_error}
We now consider the regime where we allow a catalyst to be returned with accuracy $\|\rhoinc - \rhooutc\|_1 \leq \epsilon/\log(N)$ where $N$ is the dimension of the catalyst. I.e. the error per particle is small. We will see in such a case, we recover the ordinary second law.

\subsubsection{Trivial Hamiltonians -- recovering the Shannon entropy}
\label{subsec:small_error_kids}

Let us again first consider the case of a trivial Hamiltonian. In particular, we will see that in the extensive regime only the Shannon entropy matters, and the R{\'e}nyi entropies are no longer relevant. This shows that if we relax the conditions on how cyclic the process is, by allowing relatively large inaccuracies in the returned catalyst, then we recover the usual second law.

\begin{theorem}
Let $\varepsilon \geq 0$ and let $p = {\rm spec}(\initial)$ and $q = {\rm spec}(\final)$ be the spectra of the input and output state respectively which are diagonal in the same basis and have dimension $d$. 

If there exists a catalyst with spectrum $r =r_N$ of dimension $N$ such that
\begin{equation} \label{transition}
p \otimes r \succ s, \hspace{1 cm} \Vert s - q \otimes r  \Vert_1\leq \frac{\varepsilon}{\log(N)}, 
\end{equation}
then 
\begin{equation} \label{monH}
H(p) \leq H(q) - \varepsilon - \frac{\varepsilon \log(d)}{\log(N)} - h\left(  \frac{\varepsilon}{\log(N)}  \right),
\end{equation}
with $h(x) := -x \log(x) - (1-x)\log(1-x)$ the binary entropy.

Conversely, if
\begin{equation} \label{monH2}
H(p) < H(q),
\end{equation}
then for all $N$ sufficiently large there exists a catalyst with spectrum $r =r_N$ of dimension $N$ such that
\begin{equation} \label{transition2}
p \otimes r \succ s, \hspace{1 cm} \Vert s - q \otimes r  \Vert_1  = \exp(-c\sqrt{\log(N)}).
\end{equation}
for some constant $c$.
\end{theorem}
\begin{proof}
Suppose there is a catalyst with spectrum $r$ such that Eq. (\ref{transition}) holds true. Then by Fannes inequality,
\begin{equation}
|H(s) - H(q \otimes r)| \leq \varepsilon + \frac{\varepsilon \log(d)}{\log(N)} + h\left(  \frac{\varepsilon}{\log(N)}  \right)
\end{equation}
and Eq. (\ref{monH}) follows from monotonicity of entropy under stochastic maps.

Conversely, suppose Eq. (\ref{monH2}) holds. Then we know that for all $n$ sufficiently large we have
\begin{equation} \label{majapprox}
p^{\otimes n} \succ q_n
\end{equation}
with
\begin{equation}
\Vert q_n - q^{\otimes n}  \Vert_1 \leq \exp(-c\sqrt{n}).
\end{equation}

Let us consider the following catalyst introduced in \cite{DuanFLY-multiply-catalysis}:
\begin{equation}
\label{eq:catalyst_tensor}
\omega = p^{\otimes (n-1)} \oplus p^{\otimes (n-2)} \otimes q \oplus \ldots \oplus p \otimes q^{\otimes (n-2)} \oplus q^{\otimes (n-1)}/n.
\end{equation}
We have
\begin{equation}
p \otimes \omega = p^{\otimes n} \oplus p^{\otimes (n - 1)} \otimes q \oplus \ldots \oplus p^{\otimes 2} \otimes q^{\otimes (n-2)} \oplus p \otimes q^{\otimes (n-1)}/n.
\end{equation}
Then by Eq. (\ref{majapprox}),
\begin{equation}
p \otimes \omega \succ s
\end{equation}
with
\begin{equation}
s = q_n \oplus p^{\otimes (n - 1)} \otimes q \oplus \ldots \oplus p^{\otimes 2} \otimes q^{\otimes (n-2)} \oplus p \otimes q^{\otimes (n-1)}/n.
\end{equation}
The result follows from the bound
\begin{equation}
\Vert s - q \otimes \omega \Vert_1 =  \Vert q_n - q^{\otimes n} \Vert_1 /n \leq \exp(-c\sqrt{n}),
\end{equation}
and the fact that the dimension of $\omega$ is $N = n2^{n-1}$.

\end{proof}

Note that in the above we do not have a condition that the states are diagonal in the energy eigenbasis because all states are diagonal in the energy eigenbasis of the trivial Hamiltonian.

\subsubsection{General Hamiltonians }
\label{adultHamilSmallError}
Here we prove a result analogous to the one of Section \ref{subsec:small_error_kids} in terms of the free energy and for systems with a non-trivial Hamiltonian As above, we find that if we relax the condition on how cyclic the process must be, then we recover the usual second law. Below by $p\toto q$ we mean that one can go from $p$ to $q$ by thermal operations. 

\begin{theorem}
\label{thm:unique_adults}
Let $\varepsilon \geq 0$ and let $H_S$ be the Hamiltonian. Let $p ={\rm spec}(\initial)$, $q = {\rm spec}(\final)$ be the spectra of the input and output state, respectively, which are diagonal in the energy eigenbasis. If there exists a catalyst with spectrum $r =r_N$ of dimension $N$ (with some Hamiltonian $H_C$) 
such that
\begin{equation} \label{transition_adults}
p \otimes r \toto s, \hspace{1 cm} \Vert s - q \otimes r  \Vert_1 \leq \frac{\varepsilon}{\max\{\log(N),E_{\max}\}}, 
\end{equation}
where $E_{\max}$ is maximal energy of the Hamiltonian $H_S$, then 
\begin{equation} \label{monF}
F(\rho) \geq F(\rho') - 2\varepsilon - \frac{\varepsilon \log(d)}{\log(N)} - h\left(  \frac{\varepsilon}{\log(N)}  \right),
\end{equation}
where $F$ is the standard free energy $F=E-TS$. 

Conversely, if
\begin{equation} \label{monF2}
F(p) > F(q),
\end{equation}
then for all $N$ sufficiently large there exists a catalysts with spectrum 
$r =r_N$ of dimension $N$ such that
\begin{equation} \label{transition2_adults}
p \otimes r \toto s, \hspace{1 cm} \Vert s - q \otimes r  \Vert_1 = \expbound.
\end{equation}
\end{theorem}

\begin{proof}
Suppose there is a catalyst $r$ such that Eq. \eqref{transition_adults} holds true. Then by Fannes inequality,
\begin{equation}
|H(s) - H(q \otimes r)| \leq \varepsilon + \frac{\varepsilon \log(d)}{\log(N)} + h\left(  \frac{\varepsilon}{\log(N)}  \right).
\end{equation}
Also, 
\be
|E(p) - E(q) | \leq  \varepsilon 
\ee
Therefore $|F(p) - F(q)| \leq 2\varepsilon + \frac{\varepsilon \log(d)}{\log(N)} + h\left(  \frac{\varepsilon}{\log(N)}  \right)$ which gives \eqref{monF}.

Conversely, suppose that \eqref{monF2} holds. Then from the main result of Ref. \cite{thermoiid} 
we know that  for $n$ sufficiently large
\be
p^{\ot n} \toto q_n,
\label{eq:pTOq}
\ee
with
\begin{equation}
\label{eq:qnq}
\Vert q_n - q^{\otimes n}  \Vert_1 \leq \expbound.
\end{equation}
We then consider a catalyst with the following Hamiltonian:
\be
H_C=\oplus_{k=1}^{n} \sum_{i=1}^{n} H_S^{(i)}
\ee
where $H_S^{(i)}=\overbrace{\id \ldots \ot \underbrace{H_S}_{i\text{-th site}}\ot \id \ot  \ldots \ot \id}^{n \text{terms}}$ 
and the state of catalyst of the form \eqref{eq:catalyst_tensor}.
We have
\be
p\ot\omega = \frac{1}{n} p^{\ot n} \oplus \tilde \omega,\quad \omega\ot q = \tilde\omega\oplus \frac{1}{n} q^{\ot n}
\ee
Now, let us note that  
\be
\frac{1}{n} p^{\ot n} \oplus \tilde \omega \toto  \tilde\omega\oplus \frac{1}{n} q_n.
\ee
Indeed we first apply the operation that switches the energy levels from the first term of the direct sum 
with the levels of its last term. Then we apply to the levels of the last term the operation 
that transforms $p^{\ot n}$ into $q_n$, which was shown to exists in Ref. \cite{thermoiid}, as metioned above. 
Then we reverse the above switching operation. Now, the result follows by 
\be
|| \omega\otimes q -  \tilde\omega \oplus \frac1n q_n||\leq \frac1n||q^{\ot n} - q_n ||,
\ee 
\eqref{eq:qnq}, and the fact that the dimension of the catalyst is $N = n2^{n-1}$. 
\end{proof}

\subsubsection{Small average work}
%

In Ref. \cite{skrzypczyk2013extracting} a paradigm of drawing {\it average work} was put forward. 
One starts with system in state $\rhoins$ and work register in some initial state $\rhoinw$, 
and applies thermal operations to get the work register in a new state $\rhooutw$. 
The average work is the difference of average energy of the initial and final state of work register.
This is because one can show that the protocol of drawing work in Ref. \cite{skrzypczyk2013extracting} 
is such that while repeating it with many independent systems onto the same work register,
its state becomes highly peaked around some fixed energy without degeneracies.

In \cite{skrzypczyk2013extracting} it was shown that if $\rho$ and $\sigma$ satisfy $F(\rho)\geq F(\sigma)$ 
then one can perform the transition, without spending average work (that is, such transition, 
may need sometimes a lot of deterministic work, but sometimes one gets the work back). 
It is reasonable to assume that in the case where a catalyst is used many times over many cyclic processes,
then one might not care how much work is required to return the catalyst to it's original state, but rather,
one might only care about how much work is required on average.
Indeed, if the catalyst is used many times independently, then by the law of large numbers, 
the averages work with high probability will converge to some deterministic value.

In such scenario, we are allowed to return the catalyst close to the initial catalyst, in the sense 
that by spending an $\epsilon$ amount of average work, one can go back to the original catalyst. 
According to \cite{skrzypczyk2013extracting} this means that 
\be
F(\rhoinc)\geq F(\rhooutc) -\epsilon. 
\ee

We shall  now show, that within such paradigm, if for two states we have 
$F(\rhoins)>F(\rhoouts)$ one can make transition $\rhoins \to \rhoouts$.
Indeed, suppose now that $F(\rhoins)>F(\rhoouts)$. From the previous section we know
that the following transition is possible
\be
\rhoins \ot \rhoinc \to \pi 
\ee
where $||\pi-\rhoouts\ot\rhoinc||\leq\exp(-\Omega(\sqrt(\log N))$, with $N$ being dimension of the catalyst.
Let us argue, that therefore the free energy of the states $\pi$ and $\rhoouts\ot \rhoinc$ are close.
First, the entropies are close, by Fannes inequality: 
$|S(\pi)-S(\rhoouts\ot \rhoinc)|\lesssim \log N \expbound$. Let us now estimate the energy difference. 
Applying the reasoning from section \ref{adultHamilSmallError} we have that diagonal of $\pi$ 
is given by $\tilde\omega\oplus \frac1n \sigma_n$ and diagonal of $\rhoouts\ot \rhoinc$ is $\omega \ot
 \rhoouts=\tilde\omega\oplus \rhoouts{}^{\ot n}$,
where $\sigma_n$ is a state which can be obtained from $\rhoins{}^{\ot n}$ by thermal operations (as proved in \cite{thermoiid}).
The state $\sigma_n$ satisfies  $||\sigma_n-\rhoouts{}^{\ot n}||\leq \expbound$ and both 
$\sigma_n$, and $\rhoouts{}^{\ot n}$ have Hamiltonian $\sum_{i=1}^n H_S^{(i)}$. 
Thus we have 
\be
E(\pi)-E(\rhoouts\ot \rhoinc)=\frac1n\left(E(\rhoouts{}^{\ot n})-E(\sigma_n)\right)
\leq E^{\max}_S||\rhoouts{}^{\ot n}-\sigma_n||_{\tr}\leq E^{\max}_S \expbound
\ee
Where $E^{\max}_S$ is maximal energy of a single system $S$, which is some constant, independent of the size of catalyst. 

Now, since the free energies of the real output state $\pi$, and the ideal output 
are close, then one can convert $\pi$ into $\rhoouts\ot \rhoinc$, by putting an amount of work exponentially 
small in the catalyst dimension $N$. Thus, if the customer hands to the bank the state $\pi$ 
the bank can regain the original catalyst and return the required output state of the system to the customer.

Note, that we assumed that the in order to return the catalyst to it's original state, one can act on the whole state of the system and ancilla,
while more naturally, one should act only on the catalyst. It is an interesting open question whether these two assumptions are equivalent.


\section{Proofs of properties of R{\'e}nyi entropies and divergences}\label{proofofprops}
In addition to collecting and proving useful properties of the Renyi divergences, we show in this section that by allowing error terms which are independent of the dimension of our system, the Renyi divergences  needed for our second laws can collapse to just two quantities.

\subsection{Smoothing of R{\'e}nyi entropies and divergences}


Note that the entropies and divergences are monotonic in $\alpha$ (decreasing for R{\'e}nyi entropies, and increasing for R{\'e}nyi divergences). The inequality, however, can be reversed if smoothing to a nearby state is allowed. It should be noted that  
we have explicit methods for smoothing in the regions $\alpha<1$ and $\alpha>1$, but otherwise independent of the value of $\alpha$. Nevertheless, our result give general relations in terms of so-called smoothed entropies and we will hence include them here for completeness.

\subsubsection{Definitions of smoothing}

Besides the exact R{\'e}nyi entropies, their smoothed versions have also been considered in \cite{Renner05simpleand, smoothent}, such that continuity with regard to small changes in probability distribution is preserved, and these quantities have more physical interpretations in terms of operational tasks. Their definitions are as follows:
\beq
\h^\epsilon_\alpha(p)=\left \{\bea{ll}
\displaystyle\max_{\tilde{p}}~\h_\alpha(\tilde{p}) & (\alpha < 0);\\
\displaystyle\min_{\tilde{p}}~\h_\alpha(\tilde{p}) & (0\leq\alpha\leq 1); \\
\displaystyle\max_{\tilde{p}}~\h_\alpha(\tilde{p}) & (\alpha > 1).
\eea\right.
\eeq
where optimization occurs over sub-normalized states that are $\epsilon$-close to $p$ in terms of trace distance. Mathematically, $\tilde{p}\in\mathcal{B}^\epsilon(p)$, for $\mathcal{B}^\epsilon(p) = \lbrace \tilde{p}: \half\sum_i |p_i - \tilde{p}_i| \leq \epsilon \rbrace$.

The smoothed R{\'e}nyi divergences are similarly defined, by smoothing over the $\epsilon$-ball of subnormalized states for the first argument, where maximization/minimization is taken depending on $\alpha$. Formally, 
\beq
D^\epsilon_\alpha(p\|q)=\left \{\bea{ll}
\displaystyle\min_{\tilde{p}}~D_\alpha(\tilde{p}\|q) & (\alpha < 0);\\
\displaystyle\max_{\tilde{p}}~D_\alpha(\tilde{p}\|q) & (0\leq\alpha\leq 1); \\
\displaystyle\min_{\tilde{p}}~D_\alpha(\tilde{p}\|q) & (\alpha > 1).
\eea\right.
\eeq
where optimization occurs over sub-normalized states that are $\epsilon$-close to $p$.

\subsubsection{Technical lemmas}

These quantities will be useful in considering approximate state transformations, where getting into any state close to the target state is sufficient. Also, smoothing allows us to reformulate an infinite families of conditions on the $\alpha$-R{\'e}nyi entropies to only two conditions, if states in the $\epsilon$-ball of the original state are allowed. We express this in terms of the following lemma. 


\begin{lemma}\label{rennersme}
Given any distribution $p$, for $0<\alpha<1$ and $\epsilon>0$, there exists a 
(sub-normalized) distribution $p'\in\mathcal{B}^\epsilon(p)$, such that
\begin{equation*}
H_0 (p) \geq H_\alpha (p) \geq H_0 (p') -\frac{\log\frac{1}{\epsilon}}{1-\alpha}.
\end{equation*}
For $\alpha>1$, there exists another smoothed distribution $p''\in\mathcal{B}^\epsilon (p)$, such that
\begin{equation*}
H_\infty (p'') + \frac{\log\frac{1}{\epsilon}}{\alpha-1} \geq H_\alpha (p) \geq H_\infty (p).
\end{equation*}
\end{lemma}

\begin{proof}
To prove this statement, we construct such a smoothed state $p'$ which is within the $\epsilon$-ball of $p$, and show that the inequality above holds. A construction of $p''$ can be found in \cite{rennersmoothrenyi} for $\alpha>1$. 


Note that the elements $p_i=0$ do not contribute in our calculations, furthermore the quantities we calculate are invariant under permutations of $i$. Hence, without loss of generality we can arrange them such that $p_1\leq p_2\leq \cdots\leq p_n$. For any $\ep$, denote $j$ as the maximum number such that
\begin{enumerate}
\item $\sum_{i=1}^j p_i \leq \epsilon$\\
\item $\sum_{i=1}^{j+1} p_i > \epsilon$.
\end{enumerate}
The smoothed probability distribution $p'$ is then obtained by cutting all probabilities $p_i$ where $i\leq j$. Also note that for $i\leq j$, $p_i\leq p_{j+1}$. 
Now, we evaluate a lower bound for the following quantity:
\begin{align*}
\sum_{i=1}^n p^\alpha_i &\geq \sum_{i=1}^{j+1} p^\alpha_i\geq \sum_{i=1}^{j+1} p^{\alpha-1}_i p_i \geq p^{\alpha-1}_{j+1} \sum_{i=1}^{j+1} p_i > p^{\alpha-1}_{j+1}\cdot \epsilon\geq \left[\rank(p')\right]^{1-\alpha}\cdot \epsilon.
\end{align*}
where the third inequality holds because $\alpha-1<0$, and the last inequality holds because $p_{j+1}'$ is the minimum value in the smoothed distribution $p'$, and therefore must be smaller than $\frac{1}{\rank(p')}$. Taking the logarithm and dividing by $1-\alpha>0$,
\begin{align*}
H_\alpha (p) &= \frac{1}{1-\alpha} \log \sum_{i=1}^n p^\alpha(x_i) \\
& \geq \log \rank(p')-\frac{\log \frac{1}{\epsilon}}{1-\alpha}\\
& \geq H_0 (p')-\frac{\log \frac{1}{\epsilon}}{1-\alpha}.
\end{align*}

It is worth noting that in extreme cases, where smoothing cannot be performed, the bound becomes trivial. To see this, note that the minimum value $p_1 \leq \frac{1}{n}$, where $n$ is the rank of $p$. If one cannot smooth at all, then $\epsilon<\frac{1}{n}$. The bound then translates to
\begin{equation}
H_0 (p') -\frac{\log\frac{1}{\epsilon}}{1-\alpha} \leq \log n - \frac{1}{1-\alpha}\log n \leq 0.
\end{equation}
This means that for values of $\epsilon<\frac{1}{n}$, the bound simply becomes trivially $H_\alpha (p)\geq 0$.

\end{proof}

From this lemma, combining with the fact that the smoothed entropy $H_0^\epsilon (p)$ is obtained by minimizing over all subnormalized states in the $\epsilon$-ball of $p$. More precisely, $H_0^\epsilon(p) = \min_{\tilde{p}\in\mathcal{B^\epsilon}(p)} H_0 (\tilde{p}) \leq H_0 (p')$. Hence, it is easy to obtain a corollary that corresponds to Lemmas 4.2 and 4.3 as stated in \cite{rennersmoothrenyi}.

\begin{corollary}
Given any distribution $p$, for $0<\alpha<1$ and $\epsilon>0$,
\begin{align*}
H_\alpha (p) & \geq H_0^\epsilon (p) -\frac{\log\frac{1}{\epsilon}}{1-\alpha}.
\end{align*}
For $\alpha>1$,
\begin{align*}
H_\alpha (p) & \leq H_\infty^{\epsilon} (p) + \frac{\log\frac{1}{\epsilon}}{\alpha-1}.
\end{align*}
\end{corollary}

From Lemma \ref{rennersme}, we know that the R{\'e}nyi entropies collapse to two smoothed quantities, the $\h_0(p)$ and $\h_\infty(p)$ for $0<\alpha<1$ and $\alpha>1$ respectively. In the next two lemmas, we show similar results for the smoothing of R{\'e}nyi divergences. Our proofs will be constructed in a similar way as the above proof, namely, we will specify the smoothed states within some $\epsilon$-ball such that the desired inequality holds, then link this to the smooth min- or max- divergence.

\begin{lemma}\label{divergencegeq1}
Given probability distributions $p,q$ such that ${\rm supp}(p)\subseteq{\rm supp} (q)$. Then for $\alpha>1$ and $\epsilon > 0$, there exists a smoothed distribution $p'$
\begin{equation*}
D_\alpha (p\|q) \geq D_\infty(p'\|q) - \frac{\log \frac{1}{\epsilon}}{\alpha-1}.
\end{equation*}
\end{lemma}
\begin{proof}
Note that we require ${\rm supp}(p)\subseteq{\rm supp} (q)$ so that $D_\infty (p\|q)$ does not diverge to infinity. Now, let us consider the set $Z_\delta = \lbrace i: \frac{p_i}{q_i}\geq \delta\rbrace$. Then the smoothed probability distribution $p'$ is defined by having $p_i'=\delta\cdot q_i$ for all $i\in Z_\delta$. The statistical distance between $p$ and $p'$ is
\begin{equation}\label{distbound}
d=\sum_{i\in Z_\delta} [p_i-p'_i] \leq \sum_{i\in Z_\delta} p_i,
\end{equation}
note that $d$ can be made equal to $\epsilon$, by tuning $\delta$ in a continuous manner. Then, we have
\begin{equation}
D_\infty (p'\|q) = \log \max_i \frac{p_i}{q_i} = \log\delta.
\end{equation}
We now evaluate a lower bound on 
\begin{align*}
\sum_i p^\alpha_i q^{1-\alpha}_i &\geq \sum_{i\in Z_\delta} p^\alpha_i q^{1-\alpha}_i \geq \sum_{i\in Z_\delta} p_i \cdot \left[\frac{p_i}{q_i}\right]^{\alpha-1}\geq \delta^{\alpha-1} \cdot \sum_{i\in Z_\delta} p_i \geq \delta^{\alpha-1}\cdot \epsilon,
\end{align*}
where the first inequality holds because of the positivity of $p_i$ and $q_i$ for all $i$, the third inequality holds since $\alpha-1>0$, and the last inequality holds due to \eqref{distbound}. Then, taking the logarithm, and dividing the whole equation by $\alpha-1$, we have
\begin{equation}
D_\alpha (p\|q) = \frac{1}{\alpha-1} \log \sum_x p^\alpha_i q^{1-\alpha}_i \geq \log\delta + \frac{\log \epsilon}{\alpha-1} = D_\infty (p'\|q) + \frac{\log \epsilon}{\alpha-1}.
\end{equation}
\end{proof}

\begin{lemma}\label{divergenceleq1}
Given probabilty distributions $p$ and $q$. Then for $\alpha<1$ and $\epsilon>0$, there exists a smoothed distribution $p'$ such that 
\begin{equation}
D_\alpha (p\|q) \leq D_0(p'\|q) + \frac{\log \frac{1}{\epsilon}}{1-\alpha}.
\end{equation}
\end{lemma}
\begin{proof}
Similarly in the proof of Lemma \ref{divergencegeq1}, we define a particular smoothing of $p$, called $p'$. Note also, that terms $p_{x_i}=0$ or $q_{x_1}=0$ are automatically discarded because they do not affect our calculations in any way. Firstly, we order the values $\frac{p_1}{q_1} \leq \frac{p_2}{q_2} \leq \cdots \leq \frac{p_n}{q_n} $. Subsequqently, we find an integer $j$ such that the two conditions below are satisfied:
\begin{enumerate}
\item $\sum_{i=1}^j p_i \leq \epsilon$,
\item $\sum_{i=1}^{j+1} p_i > \epsilon$.
\end{enumerate}
Note that $D_0 (p'\|q)=-\log\sum_{i=j+1}^n q_i$. Also, for any $i\geq j+1$, $\frac{p_i}{q_i}\geq\frac{p_{j+1}}{q_{j+1}}$. 

With this, we can evaluate the following bound:
\begin{align*}
\sum_{i=j+1}^n q_i &\leq \frac{q_{j+1}}{p_{j+1}}\cdot \sum_{i=j+1}^n p_{i} \leq \frac{q_{j+1}}{p_{j+1}}.
\end{align*}
subsequently, we can evaluate a lower bound on the following summation:
\begin{align*}
\sum_{i=1}^n p_i^\alpha q_i^{1-\alpha} &\geq \sum_{i=1}^{j+1} p_i^\alpha q_i^{1-\alpha}\\
& \geq \sum_{i=1}^{j+1} p_i \cdot \left[\frac{p_i}{q_i}\right]^{\alpha-1}\\
& \geq \left[\frac{p_{j+1}}{q_{j+1}}\right]^{\alpha-1} \cdot \sum_{i=1}^{j+1} p_i \\
& > \left[\frac{p_{j+1}}{q_{j+1}}\right]^{\alpha-1} \cdot \epsilon\\
& \geq \left[\frac{1}{\sum_{i=j+1}^n q_i}\right]^{\alpha-1} \cdot \epsilon,\\
\end{align*}
where the second and fifth inequality holds since $\alpha-1<0$, and others by srtaightforward manipulation.

It is useful to see that even for the case where $j=0$ where no smoothing is possibly done, the proof still holds since now we know that $\epsilon<p_1$, where $p_1$ is the value that corresponds to the minimum in $\frac{p_i}{q_i}$. Note that in this case, 
\begin{equation}
\sum_{i=1}^n p_i^\alpha q_i^{1-\alpha} \geq \left[\frac{p_1}{q_1}\right]^{\alpha-1} p_1 \geq \epsilon.
\end{equation}
This comes from two facts: $p_1 > \epsilon$, and the fact that $\frac{p_1}{q_1}\leq 1$. To see why this is the case, let us prove by contradiction, i.e. assume that $\frac{p_1}{q_1}>1$ and all other $\frac{p_i}{q_i}>1$ also. Then for each $i$, $p_i>q_i$, and $\sum_{i=1}^n p_i > \sum_{i=1}^n q_i =1$ which is impossible.
\end{proof}

With these two lemmas, two simple corollaries with regard to smooth divergences can be obtained:
\begin{corollary}
Given distributions $p$ and $q$. Then for $\alpha>1$ and $\epsilon> 0$,
\begin{equation}
D_\alpha (p\|q) \geq D_\infty^\epsilon(p\|q) - \frac{\log \frac{1}{\epsilon}}{\alpha-1}.
\end{equation}

\end{corollary}
\begin{corollary}
Given distributions $p$ and $q$. Then for $\alpha<1$ and $\epsilon>0$,
\begin{equation}
D_\alpha (p\|q) \leq D_0^{\epsilon}(p\|q) + \frac{\log \frac{1}{\epsilon}}{1-\alpha}.
\end{equation}
\end{corollary}

We have also investigated the R{\'e}nyi entropies for the $\alpha<0$ regime, and shown that they are monotonically increasing in $\alpha$. Although they are not required for our work upon elimination of negative $\alpha$, we also state them here for completeness.
\begin{theorem}[Jensen's inequality]\label{jensen}
For any convex function $f$, the following inequality holds:
\begin{equation}
f\left(\frac{\sum_i a_i y_i}{\sum_i a_i}\right) \leq \frac{\sum_i a_i f(y_i)}{\sum_i a_i}.
\end{equation}
The inequality is reversed for concave functions.
\end{theorem}

Using Theorem \ref{jensen}, we will prove the following lemma about the monotonicity of $\h_\alpha(p)$ in the negative $\alpha$ regime. 
\begin{lemma}
For $\forall\alpha'\leq \alpha<0$, $\h_{\alpha'}(p)\leq \h_{\alpha}(p)$ for any probability distribution $p$.
\end{lemma}
\begin{proof}
For $\alpha<0$ and probability distribution $p$, $\h_\alpha(p) = \frac{1}{\alpha-1}\log\sum_i p^\alpha$. For convenience of dealing with positive numbers always, we rewrite this expression by defining a variable $\kappa=-\alpha$, and for $\kappa \in (0,\infty)$, define
\begin{equation}
\hat{\h}_\kappa (p) = \frac{-1}{1+\kappa}\log\sum_i x_i^\kappa= \h_\alpha (p),
\end{equation}
where $x_i=\frac{1}{p_i}$ for all $i$. We then prove that $\frac{\partial \hat{\h}_\kappa (p)}{\partial  \kappa} \leq 0$.
\begin{align}\label{delh}
\frac{\partial \hat{\h}_\kappa (p)}{\partial  \kappa} & = \frac{1}{(1+\kappa)^2}\log\sum_i x_i^\kappa - \frac{1}{1+\kappa} \frac{1}{\sum_i x_i^\kappa} \sum_i x_i^\kappa \log x_i\nonumber\\
& = \frac{1}{(1+\kappa)^2}\frac{1}{\sum_i x_i^\kappa} \left[ \sum_i x_i^\kappa\log\left(\sum_j x_j^\kappa\right) - (1+\kappa) \sum_i x_i^\kappa \log x_i \right]\nonumber\\
& = \frac{1}{(1+\kappa)^2}\frac{1}{\sum_i x_i^\kappa} \left[ \sum_i x_i^\kappa\log\left(\sum_j x_j^\kappa\right) - \sum_i x_i^\kappa \log x_i^{(1+\kappa)} \right].
\end{align}
By defining the function $f(x)=x\log x$ which is convex in $\mathcal{R}^+$, we can apply Theorem \ref{jensen} by setting $a_i=x_i^{-1}=p_i$ and $y_i=x_i^{1+\kappa}$. This implies that 
\begin{equation}
\left[ \sum_i x_i^\kappa\log\left(\sum_j x_j^\kappa\right) - \sum_i x_i^\kappa \log x_i^{(1+\kappa)} \right] \leq 0
\end{equation}
and hence \eqref{delh} is upper bounded by 0 due to the positivity of the first two terms.
\end{proof}

\subsection{Smoothing relations}\label{smoothingrelations}
\label{sec:collapseCondition}
It turns out that our infinite set of conditions can often be verified by checking just two conditions in an approximate sense. This applies to the case where only the conditions for $\alpha \geq 0$ are relevant. However, as we will argue later this is generally sufficient.

Let us first explain how this works by considering only the R{\'e}nyi 
entropies, which are the relevant quantities when the Hamiltonian is trivial. Note that if $\alpha \geq \beta$ then for all distributions $p$ we have $H_{\alpha}(p) \leq H_{\beta}(p)$. The key to approximately reducing the number of conditions is to note that there exists a distribution quite close to $p$ such that up to some error terms the entropies can also be related in the opposite direction. Closeness is thereby measured in terms of the statistical distance and we use $\mathcal{B}^\epsilon(p) = \lbrace p': \half\sum_i |p_i - p'_i|\rbrace$ to denote the $\epsilon$ ball of (sub-normalized) distributions $p'$ around $p$. We will also call such a 
$p'$ a \emph{smoothed} distribution~\cite{Renner05simpleand, smoothent}.

Specifically, we will show (see Lemma~\ref{rennersme}) that for any $0 < \alpha < 1$, any distribution $p$ and any $\epsilon > 0$, there exists a 
smoothed distribution $p' \in \mathcal{B}^{\epsilon}(p)$ such that
\begin{equation}
H_0 (p) \geq H_\alpha (p) \geq H_0 (p') -\frac{\log\frac{1}{\epsilon}}{1-\alpha}.
\end{equation}
Similarly, whenever $\alpha>1$ there exists another distribution $p''\in\mathcal{B}^\epsilon (p)$ such that
\begin{equation}
H_\infty (p'') + \frac{\log\frac{1}{\epsilon}}{\alpha-1} \geq H_\alpha (p) \geq H_\infty (p).
\end{equation}
This means that whenever we demand that $H_\alpha (\initial)\leq H_\alpha (\final)$ for all values of $\alpha \geq 0$, we can reduce the set of conditions in an approximate sense by relating $H_{\alpha}$ to $H_0$ or $H_{\infty}$. 
More precisely, given probability distributions $p$ and $q$ and $\epsilon>0$
we can construct smoothed distributions $p' \in \mathcal{B}^{\epsilon}(p)$ 
and $q'' \in \mathcal{B}^{\epsilon}(q)$ according to explicit smoothing strategies as in Lemma \ref{rennersme}. If the following conditions are satisfied
\begin{itemize}
\item For $0<\alpha<1$, $H_0 (p') - \frac{\log\frac{1}{\epsilon}}{1-\alpha}\geq H_0 (q)$
\item For $\alpha>1$, $H_\infty (p) \geq H_\infty (q'') + \frac{\log\frac{1}{\epsilon}}{\alpha-1}$,
\end{itemize}
then $\forall \alpha>0$, it holds that $H_\alpha(p)\geq H_\alpha (q)$. 
As we will see these conditions can also be expressed in terms of smoothed entropies~\cite{Renner05simpleand}, however, we would like to emphasize that there are in fact only two smoothing strategies, one for $\alpha \geq 1$ and 
one for $\alpha < 1$. This means that one could apply these smoothing strategies, and only verify the two conditions stated above. It should also be emphasized that this allows a verification in one direction only. Namely, if we find that the conditions above are satisfied, then we can conclude that also the original conditions are satisfied for all $\alpha$. However, due to the approximations above the converse does not hold in the sense that the original conditions may be satisfied and yet the fudge terms in the conditions above no longer allow for a verification.

A similar statement can be made for the R{\'e}nyi divergences, which are relevant for the case of a non-trivial Hamiltonian.
Here we want to check whether given initial and final states $\initial$ and $\final$ we have that for all $\alpha \geq 0$, $D_\alpha(\initial\|\gibbs)\geq D_\alpha(\final\|\gibbs)$. 
Again, we have a bound in one direction as the R{\'e}nyi divergences are monotonically increasing in $\alpha$, hence $D_\alpha(\final\|\gibbs)\geq D_\beta (\final\|\gibbs)$ whenever $\alpha \geq \beta$~\cite{Tomamichel-thesis}. 
As we show below, one can again obtain a bound in the other direction by considering smoothed distributions. 
We show (Lemma \ref{divergencegeq1}), that for any distribution $p$ 
there exists a particular smoothed distribution $p' \in \mathcal{B}^{\epsilon}(p)$ such that $D_\alpha (p\|q)\geq D_\infty (p'\|q) - c$ for all $\alpha>1$. $c$ is a logarithmic fudge factor that depends on the smoothing parameter. A similar statement holds for $\alpha < 1$, by relating $D_\alpha$ to $D_0$ (see Lemma~\ref{divergenceleq1}). Again, the smoothed distributions only depend on $\epsilon$ and whether $\alpha >1$ or $\alpha < 1$.

These relations now again allow us to simplify our conditions in an approximate sense. Consider probability distributions $p$, $q$ and $r$ where $q$ has full rank $\rank(q)=n$.
For $\epsilon>0$, apply the $\epsilon$-smoothing strategy in proof of Lemma \ref{divergencegeq1} for $p$ to obtain $p' \in \mathcal{B}^\epsilon(p)$. Then if
\begin{align}
D_\infty (p'\|r) &\gtrsim D_\infty (q\|r),
\end{align}
we have $D_\alpha (p\|r)\geq D_\alpha (q\|r)$ for $\alpha>1$.
Similarly, apply the $\epsilon$-smoothing strategy in Lemma \ref{divergenceleq1} for $q$ to obtain $q' \in \mathcal{B}^{\epsilon}(q)$. 
Then if the following conditions are satisfied
\begin{align*}
D_0 (p\|r) &\gtrsim D_0(q'\|r),
\end{align*}
then $\forall \alpha>0$, it holds that $D_\alpha (p\|r)\geq D_\alpha (q\|r)$.

\section{Comparison to other models}
\label{sec:comparison}
\subsection{Thermal operations and time dependent Hamiltonians}
\label{sec:change_hamil}
\label{ss:changingham}
 
 As shown in \cite{HO-limitations}, one can incorporate changing Hamiltonians into thermal operations by introducing an ancillary system.
 Now that we have conditions for state transformations using ancillary systems we are in a position to prove the optimality of this procedure.

In the standard thermodynamics, one usually deals with time dependent Hamiltonian. 
For example, expanding a gas in a container by drawing a piston (and thereby changing the Hamiltonian,
one goes from the thermal state of volume $V_1$ to the thermal state of volume $V_2$). This results in obtaining work. 

In the paradigm of thermal operations, we have been discussing the case of a fixed Hamiltonian. 
There is a possibility of changing the Hamiltonian of the system in the following sense:
given a system in state $\rho$ with Hamiltonian $H$, one brings in another system in the thermal state with 
Hamiltonian $H'$. Then one can apply thermal operations to this compound system, 
and finally trace out the initial system, obtaining some output state $\sigma$ with Hamiltonian $H'$. 
In particular, in this way, one can go for free from the thermal state with one Hamiltonian to a thermal state 
with another Hamiltonian, which may seem to contradict a common thermodynamical paradigm, 
where changing the Hamiltonian is related to performing work. However, the above "change" of Hamiltonian 
is not really a change: the Hamiltonian of the universe is fixed, and  we simply have turned to another system
with the required new Hamiltonian, and removed the system with the old Hamiltonian. 
Thus in  this case, while going from the thermal state $(\rho_\beta,H)$ to another thermal state $(\rho'_\beta,H')$
one does not gain any work, nor does one need any work while performing such an operation.

In standard thermodynamics, in the above transition, 
work is performed equal to the difference of free energy of the initial and final state. 
The way to reconcile the paradigm of thermal operations with the more common picture 
of changing the Hamiltonian was given in \cite{HO-limitations} (see also \cite{thermoiid}). Suppose one wants to go from $\rho$ 
to $\sigma$ and change Hamiltonian form $H$ to $H'$, however not in the trivial way described above,
which does not cost any work, but in a traditional sense. 
As noted, the total Hamiltonian of the universe is constant,  so such a change of Hamiltonian means 
that we actually use some other system, usually a clock system,  such that the total Hamiltonian 
is 
\be
\sum_{t=t_{i}}^{t=t_f} H(t) \ot |t\>\<t |
\ee
where $|t\>$ are  orthogonal states.
Then in more standard thermodynamics, we set the initial state to be $\rho \ot |t_{i}\>\<t_{i}|$, and let such a system (but
not the clock) weakly interact with the heat bath.  
Due to the weakness of interactions, during the system's evolution, the state is approximately in product form $\rho(t) \ot |t\>\<t|$ 
Then  the final state is $\sigma \ot |t_{f}\>\<t_{f}|$, and we say we have changed state $\rho$ to 
$\sigma=\rho(t_f)$ and $H(t_{i})$ into $H(t_{f})$.  
The work is then equal to $\Delta W = \int_{t_i}^{t_f} \tr H(t) \rho(t)$. 

In our present approach, the weak interaction with the heat bath is replaced by a unitary transformation,
that commutes with the total Hamiltonian. We require only, that the final state is of product form $\rho(t_f)\ot |t_f\>\<t_f|$.
The state $t_f$ we take orthogonal to $t_{i}$. However, now one can greatly simplify it, as for ideal processes, the possibility 
of a transition does not depend on the intermediate times, and one can use just two values of $t$: $|t_{i}\>\equiv |0\>$
 and $|t_f\>\equiv |1\>$.
 
Therefore, to mimic the change of Hamiltonian, we consider a fixed Hamiltonian 
\be
 H= H_0 \ot |0\>\<0| + H_1 \ot |1\>\<1|
 \label{eq:H01}
\ee
where we have an additional system, that plays the role of a switch bit, and its state changes from $|0\>$ to $|1\>$, so that we have
a transition between $\rho \ot |0\>\<0|$ and $\sigma \ot |1\>\<1|$ and Hamiltonian $H_0$ to $H_1$ acting on the system.
As shown in \cite{HO-limitations}, if $\rho$ and $\sigma$ are thermal states 
then the amount of deterministic work we need to perform is equal precisely to the 
difference of standard free energies. In this picture, the variety of paths that might 
lead from $H_0$ to $H_1$ is replaced by the variety of thermal operations that 
may be applied to the initial state, and give the final state $\sigma \ot |1\>\<1|$.

As an example, let us observe, that by using the paradigm of thermal operations, for $\rho$ and $\sigma$ 
being thermal states, we will obtain the same answer as in the traditional paradigm: the amount of work
is equal to the difference of free energies. 

We start with $(\gibbsin\ot |0\>\<0|\ot |E\>_W\<E|,H)\equiv(\rho_{\rm in},H)$ and want to end up with 
$(\gibbsout\ot |1\>\<1|\ot |E'\>_W\<E'|,H)\equiv(\rho_{\rm out},H)$  where $\gibbsin$ and $\gibbsout$ are 
thermal states for Hamiltonians $H_0$ and $H_1$, respectively. 
The total Hamiltonian is of the form
\be
H=H_0 \ot |0\>\<0| + H_1 \ot |1\>\<1| + H_W
\ee
where $H_W$ is the Hamiltonian of the work system. 
One then computes that 
$k T D_\alpha(\rho_{\rm in})=kT \ln (Z_0/Z) + E$  and $k T D_\alpha(\rho_{\rm out})=kT \ln (Z_1/Z) + E'$ 
where $Z_0$, $Z_1$ are partition functions for $H_0, H_1 $ in temperature $T$. 
Thus the transition is possible, whenever $F(\gibbsin)-F(\gibbsout)\geq \Delta W$,
where $\Delta_W=E'-E$ is the amount of performed work.

For general  states, we obtain, that $(\rho \ot |0\>\<0|, H)$ 
can be transformed into $(\sigma \ot |1\>\<1|, H)$, with $H$ of the form \eqref{eq:H01}
if and only if the {\it generalized free energies} satisfy
\be
F_\alpha (\rho,H_0) \geq F_\alpha(\sigma,H_1)
\label{eq:alphaFmono}
\ee
where $F_\alpha$ is defined as follows, for any given state and Hamiltonian:
\be
F_\alpha(\rho,H)= -k T \ln Z + kT D_\alpha(\rho\|\rho_\beta)
\ee
where $\rho_\beta$ is the thermal state for the Hamiltonian $H$.

Finally, one can ask, whether the Hamiltonian of the form  \eqref{eq:H01} is optimal for the transition 
from $(\rho_0, H_0) $ to $(\rho_1,H_1)$? 

More generally one could do the following. 
We consider initial state $\rho_0\ot\sigma_0$ and $\rho_1 \ot \sigma_1$,
where $\sigma_0$, and $\sigma_1$ are arbitrary states of ancillas (replacing the switch states $|0\>$ and $|1\>$). 
Then we are to choose an arbitrary total Hamiltonian which satisfies the following conditions: 
for any state from the support of $\rho_0\ot\sigma_0$ it acts as $H_0$, while for any state 
with support  $\sigma_0$, and $\sigma_1$ it acts as $H_1$. We also demand that $\sigma_0$ 
and $\sigma_1$ have to be related unitarily, otherwise, going from $\sigma_0$ and $\sigma_1$ 
can be exploited  as a resource(if e.g. $\sigma_0$ is pure, while $\sigma_1$ is not).

Now one computes that 
\ben
D_\alpha(\rho_0\ot \sigma_0, H)=-\ln Z_0/Z + D_\alpha(\rho_0\|\gibbsin)+D_\alpha(\sigma_0\|I/d) \nonumber\\
D_\alpha(\rho_0\ot \sigma_0, H)=-\ln Z_1/Z + D_\alpha(\rho_1\|\gibbsout)+D_\alpha(\sigma_1\|I/d) \nonumber\\
\een

Using that $D_\alpha(\sigma_0\|I/d) =D_\alpha(\sigma_1\|I/d)$ we obtain 
that  the condition of monotonicity of $D_\alpha$  is equivalent to \eqref{eq:alphaFmono}. 
Thus the simple switch mechanism of Hamiltonian \eqref{eq:H01} is optimal.

\subsection{Universality of the work bit}

Although we have used a specific form of work system to invest/extract work (the work bit), it was claimed to be equivalent to other work systems in \cite{HO-limitations}. 
However, as the work system is an ancilla, one needs to check that this universality continues to hold in the context where catalysts are allowed. One wants to show that the derived work distance is general, i.e. whether by considering other proposed forms of battery systems, we arrive at the same quantity. In fact, we will see that the second laws impose
a constraint on how one defines work, if one demands reversibility of the work system. We see that this is satisfied 
when work is defined in terms of raising and lowering the energy of a pure state (or a system with highly peaked energy). This is because transitions between pure
states are not effected by catalysis. However, in models where the Hamiltonian is trivial, i.e. $H=0$, one needs an alternative way to define work. 
For example, one can define it in terms of the number of maximally mixed states which are erased into pure states, as was done in \cite{faist2012quantitative}. 
In this case, we can again apply our second laws to the initial and final state, including the battery, to derive an upper bound on $W$.

More generally, let us take the initial state of the work system $W$ to be $W_i$ and the final state to be $W_f$. Considering the conditions we derived for state transformation on the joint system $SW$, from $\initial$ to $\final$ on system $S$, $\forall \alpha\geq 0$ (since the probability distributions have zeros, conditions for $\alpha\<0$ become redundant), we have\begin{align}
F_\alpha (\final\otimes W_f\|{\gibbs}_{SW}) &\leq F_\alpha (\initial\otimes W_i\|{\gibbs}_{SW}) 
\end{align}
This implies
\begin{align}
F_\alpha(W_f\|{\gibbs}_W)-F_\alpha(W_i\|{\gibbs}_W)&\leq F_\alpha (\initial\|{\gibbs}_S) -F_\alpha (\final\|{\gibbs}_S)
\label{eq:battery}
\end{align}

It is reasonable that for any definition of work system and work extraction, we require that one can reversibly go from state $W_i$ to $W_f$ and visa versa. This is equivalent to requiring that 
\begin{align}
{\tilde W}=F_\alpha(W_f\|{\gibbs}_W)-F_\alpha(W_i\|{\gibbs}_W) \,\,\, \forall \alpha
\end{align}
This is because if $F_\alpha(W_f\|{\gibbs}_W)-F_\alpha(W_i\|{\gibbs}_W)$ were different for two values of $\alpha$, then it would require a different amount of pure energy to go from
the initial state to the final state, than from the final state to the initial.  We can then just define this constant ${\tilde W}$ to be the amount of work $W$,
and thus the amount of work obeys $W=F_\alpha(W_f\|{\gibbs}_W)-F_\alpha(W_i\|{\gibbs}_W)$ for all $\alpha$.

Now, in the case of close to pure energy states, $F_\alpha(W_f\|{\gibbs}_W)-F_\alpha(W_i\|{\gibbs}_W)=E_f-E_i$ and so we recover that the 
work is just the change in energy of the work system, regardless of what type  of system it is.
For the battery used in \cite{faist2012quantitative} this is also the case: denote the work system $W$ as consisting of $n$ qubits, where the Hamiltonian $\hat{H}_W= \id$ is trivial. $n$ can be arbitrarily large. Its initial state is described as $W_i = |0\>\<0|\otimes 2^{-\lambda_1}\id$. Physically this just implies that the state consists of $\lambda_1$ maximally mixed qubits, and the remaining $n-\lambda_1$ qubits are pure. Its final state is similarly defined as $W_f=|0\>\<0|\otimes 2^{-\lambda_2}\id$. The thermal state of system W is ${\gibbs}_W=2^{-n}\id$, which is maximally mixed.

What is interesting about this comparison is that while in our model, work is stored in the form of energy, in \cite{uniqueinfo,faist2012quantitative} the work is quantified by means of purity, i.e. how many pure qubits we invest/create during the process of state transformation. The significance of information to work extraction has been discussed by various works 
\cite{uniqueinfo,360221697,dahlsten2011inadequacy}. In particular, Landauer's principle \cite{Landauer,Keyes:1970:MED:1662935.1662942} states that any physical process that erases one bit of information (i.e., creating purity) in an environment of temperature $T$ has a fundamental average work cost of $kT\ln(2)$. Similarly, by ulitizing one bit of information stored (i.e., consuming purity) in a physical system, and allowing it to interact with a thermal bath at temperature $T$, one can draw an average work of $kT\ln(2)$.

We will now use this battery system $W$.  
From Equation \eqref{eq:battery} we can derive an upper bound for the quantity $\lambda=\lambda_1+\lambda_2$ 
\begin{align}
kT(\lambda_1 -\lambda_2 )\leq
F_\alpha (\initial|{\gibbs}_S)-F_\alpha (\final\|{\gibbs}_S) 
\end{align}
Since the work is defined via the process of Landauer erasure as $kT(\lambda_1 -\lambda_2)$, we see that this bound is equivalent to the bounds in \eqref{workdistderivation}. The quantity $\lambda=\lambda_1 - \lambda_2$ denotes the \textit{net gain of pure qubits} in the process. By Landauer's principle $\lambda$ gives a bound on the maximum amount of work extractable, as no more than $kT\ln(2)$ amount of work can be extracted given one bit of pure information. We thus see that our condition yields
the same upper bound on $W$
\begin{align*}
W &\leq \inf_{\alpha>0}[F_\alpha (\initial\|{\gibbs}_S)-F_\alpha (\final\|{\gibbs}_S)]\\
\end{align*}
There are, however, a few more subtle differences between these two models for work systems. For instance, note that this battery consists of qubits, and hence $\lambda_1$ and $\lambda_2$ take integer values. \cite{faist2012quantitative} has shown that they can be further generalized to take values of rational numbers, however, one has to postulate 
the connection between work and information. In our wit model, however, $W$ can take any value, including irrational values.


{\bf Acknowledgements} 

We thank Robert Alicki, Piotr Cwiklinski, Milan Mosonyi,  Sandu Popescu, Joe Renes, Marco Tomamichel and Andreas Winter for useful discussions, and Max Frenzel for comments on our draft. JO is supported by the Royal Society. MH is supported by the Foundation for Polish Science TEAM
project cofinanced by the EU European Regional Development Fund.
NN and SW are supported by the National Research Foundation and Ministry of Education (MOE), Singapore as well as MOE Tier 3 Grant "Random numbers from quantum processes" (MOE2012-T3-1-009).

\end{document}